\pdfoutput=1 
\documentclass[webpdf,modern,large]{two-cols-template}
\PassOptionsToPackage{dvipsnames}{xcolor}

\onecolumn 

\usepackage{setspace}
\graphicspath{{../figs/}}

\usepackage[left=1in, right=1in,bottom=1in,top=1in]{geometry}
\RequirePackage[dvipsnames]{xcolor}
\usepackage{tikz}
\usetikzlibrary{fit}
\usepackage{epstopdf}
\usetikzlibrary{positioning}
\tikzset{mynode/.style={draw,text width=1.5in,align=center}}
\tikzset{mynodenoborder/.style={text width=1in,align=center}}
\usetikzlibrary{calc}
\usepackage{cleveref}
\usepackage[normalem]{ulem}
\usepackage{lipsum,parskip}
\usepackage{svg}
\usepackage{hyperref}
\usepackage{booktabs}
\hypersetup{colorlinks,linkcolor={blue},citecolor={blue},urlcolor={blue}}
\usepackage{xr}
\usepackage{soul}
\usepackage{placeins}

\usetikzlibrary{shapes.geometric, arrows, positioning, calc, matrix}
\tikzset{mynode/.style={draw,text width=1.5in,align=center}}
\tikzset{mynodenoborder/.style={text width=1in,align=center}}
\usetikzlibrary{arrows.meta, positioning}
\tikzset{
  box/.style = {
    rectangle, rounded corners, minimum width=5cm, minimum height=2cm,
    text centered, draw=black, fill=white},
  arrow/.style={very thick, -Stealth},
  header/.style={
    label={[rectangle, fill=white, draw, anchor=center,
            minimum width=2cm, node font=\ttfamily,
            name=\tikzlastnode-header]north:{#1}}}}

\definecolor{c1}{HTML}{F8766D}
\definecolor{c2}{HTML}{00BA38}
\definecolor{c3}{HTML}{619CFF}

\newcommand{\tblue}[1]{\textcolor{blue}{#1}}

\newcommand{\argmin}{\operatornamewithlimits{argmin}}
\newcommand{\argmax}{\operatornamewithlimits{argmax}}

\newcommand{\bB}{\bold{B}}
\newcommand{\bE}{\bold{E}}

\newcommand{\be}{\bold{e}}

\newcommand{\bY}{\bold{Y}}
\newcommand{\bX}{\bold{X}}

\newcommand{\bU}{\bold{U}}

\newcommand{\bGamma}{\bold{\Gamma}}

\newcommand{\bF}{\bold{F}}
\newcommand{\bS}{\bold{S}}
\newcommand{\bL}{\bold{L}}

\newcommand{\bR}{\bold{R}}
\newcommand{\bI}{\bold{I}}

\newcommand{\Cov}{\bf{\text{Cov}}}
\newcommand{\Cor}{\bf{\text{Cor}}}

\newcommand{\bPi}{\boldsymbol{\Pi}}
\newcommand{\bSig}{\boldsymbol{\Sigma}}

\newcommand{\bfGamma}{\boldsymbol{\Gamma}}

\newcommand{\tr}{\operatornamewithlimits{Trace}}

\newcommand{\bfSigma}{\boldsymbol{\Sigma}}

\newcommand{\indep}{\operatorname{\perp\!\!\!\perp}}
\definecolor{myred}{HTML}{FF1F5B}
\definecolor{myblue}{HTML}{009ADE}
\definecolor{mygreen}{HTML}{00CD6C}
\definecolor{myyellow}{HTML}{FFC61E}
\definecolor{mypurple}{HTML}{AF58BA}

\usepackage{algpseudocode}
\makeatletter
\newcommand{\multiline}[1]{%
    \begin{tabularx}{\dimexpr\linewidth-\ALG@thistlm}[t]{@{}X@{}}
        #1
    \end{tabularx}
}
\newcommand{\Step}[1]{\algrenewcommand{\alglinenumber}[1]{\textbf{Step ##1:} } #1}

\def\NoNumber#1{{\def\alglinenumber##1{\ \ \ \ \ \ \ \ \ \ }\State #1}\addtocounter{ALG@line}{-1}}

 \algnewcommand\Input{\item[\ \ \ \ \textbf{Input:}]}
 \algnewcommand\Output{\item[\ \ \ \ \textbf{Output:}]}
\algnewcommand\algorithmicforeach{\textbf{for each}}
\algdef{S}[FOR]{ForEach}[1]{\algorithmicforeach\ #1\ }

\makeatletter
\newcommand{\Suppressnumber}{\def\alglinenumber##1{\ \ \ \ \ \ \ \ \ \ }}
\newcommand{\Reactivatenumber}[1]{\setcounter{ALG@line}{\numexpr#1-1}}
\newcommand{\NewStep}[1]{\algrenewcommand{\alglinenumber}[1]{\textbf{Step ##1:} } #1}
\makeatother

\theoremstyle{thmstyleone}%
\newtheorem{proposition}{Proposition}
\newtheorem{lemma}{Lemma}%

\begin{document}

\journaltitle{arXiv Perprint}
\copyrightyear{2025}
\pubyear{2019}
\appnotes{Article}

\firstpage{1}


\title[CoReg]{Modeling Dependence in Omics Association Analysis via Structured Co-Expression Networks to Improve Power and Replicability} 
\author[1,2,3]{Hwiyoung Lee\ORCID{0000-0002-3855-2316}}
\author[1,2,3]{Yezhi Pan\ORCID{0009-0001-4014-0302}}
\author[1,2,3]{Shuo Chen\ORCID{0000-0002-7990-4947}}

\address[1]{\orgdiv{Epidemiology \& Public Health},\orgname{University of Maryland, School of Medicine}, \orgaddress{\street{Baltimore}, \postcode{MD}, \country{USA}}}
\address[2]{\orgdiv{Maryland Psychiatric Research Center (MPRC)},\orgname{University of Maryland, School of Medicine}, \orgaddress{\street{Catonsville}, \postcode{MD}, \country{USA}}}
\address[3]{\orgdiv{Institute for Health Computing},\orgname{University of Maryland}, \orgaddress{\street{North Bethesda}, \postcode{MD}, \country{USA}}}


\received{Date}{0}{Year}
\revised{Date}{0}{Year}
\accepted{Date}{0}{Year}

\abstract{Accounting for dependence among high-dimensional variables in omics data analysis is critical to obtain accurate and reliable statistical inference. Although latent, omics variables often exhibit structured correlation/co-expression patterns. However, there are few methods explicitly accounting for such structured dependence in the statistical analysis of omics data (e.g., differential expression analysis). To address this methodological gap, we propose a \textbf{Co}expression network multivariate \textbf{Reg}ression (CoReg), which integrates co-expression network structure into multivariate regression analysis to precisely account for the inter-correlations (dependence) among omics variables. We show in simulations that CoReg substantially improves the accuracy of statistical inference and replicability across studies. These findings suggest that CoReg provides an alternative approach for omics data association analysis with dependence adjustment, analogous to the role of mixed-effects models in handling repeated measures in lower-dimensional settings.}
\maketitle

\doublespacing

Multivariate association analysis has long served as a fundamental approach for linking large-scale molecular signals (e.g., epigenomics, transcriptomics, proteomics, and metabolomics) to biological traits of interest. For instance, as a common first step in diagnostic biomarker discovery, differential expression (DE) analysis is used to identify omics features that are differentially expressed between cancer cases and controls. Such multivariate analyses are typically followed by multiple-testing correction procedures to control the rate of false-positive findings.

From a statistical perspective, properly accounting for dependence is central to multivariate association analysis, as it improves statistical efficiency, reduces the Type I error rate, and enhances reproducibility and replicability  \citep{Yu:2022(Neuron),Mukamel:2025(Nature)}. For example, mixed-effects models have been widely applied to address dependence among multivariate outcomes in longitudinal clinical trials. Nevertheless, accurately modeling dependence structures in high-dimensional omics outcomes remains a long-standing challenge. 

High-dimensional omics variables derived from complex and highly organized biological systems often exhibit strong correlations that follow organized network topological structures \citep{Langfelder2008}. Co-expression network analysis methods have been successfully applied to uncover these dependencies from covariance and correlation matrices \citep{Qiong:2021ICN}, with many findings validated by known biological pathways. Nonetheless, integrating such structured network dependencies into multivariate statistical inference remains a major methodological gap due to the ultra-high dimensionality of parameters in the covariance matrix and complex network topology. 

Due to aforementioned challenges, mass-univariate approaches are commonly used for omics data analysis, which involve performing separate regressions (or two-sample tests) for each omics variable followed by multiple testing correction. This procedure can be problematic, as they ignore dependencies among variables, potentially leading to reduced statistical efficiency on individual omics variables and violated assumptions for subsequent multiple testing correction (e.g., independence of tests for \citep{BHFDR:1995}). Consequently, missing the omics variables truly associated with the health condition and false positively identifying non-related omics variables jointly contribute to reduced replicability of findings across studies \citep{Perng:2020replication}. These replicability challenges are a widespread problem in scientific research \citep{Mukamel:2025(Nature), Singh:2023, Baker:2016,Ioannidis:2005}. 

To address this issue, various advanced statistical methods have been developed. For example, the low dimensional factor guided regression models have been well-established for modeling the unstructured dependence patterns among high-dimensional outcomes \citep{Leek(PNAS):2008,Fan:2024}, and for adjusting confounding and unmodeled effects \citep{Wang:2017, Bing:2024}, while other methods combined the co-expression and association analysis for module-level inference \citep{Tesson:2010, Dam:2017}. However, these models remain limited to incorporating the latent yet organized network structures of the large covariance matrix  (see \Cref{fig:omics}) into modeling the dependence/covariance in the multivariate regression model of omics variables. Therefore, these methods may not fully account for the organized dependence patterns of omics data. This limitation reduces statistical inference accuracy, leading to decreased sensitivity and increased false discovery rates (FDR), and thus less replicable findings.

To fill this methodological gap, we propose a \textbf{Co}-expression network multivariate \textbf{Reg}ression (CoReg) framework. Building on the principles of low-dimensional modeling, CoReg explicitly models the covariance matrix of multivariate omics variables by leveraging the extracted network structures from co-expression analysis and incorporating them into a network-guided factor analysis. 
By modeling structured dependence structures in multivariate regression models, the CoReg can simultaneously improve the accuracy of statistical inference on individual omics variables and subsequent multiple testing correction, and consequently the replicability.

\begin{figure*}[h]
    \centering 
    \includegraphics[width=0.9\linewidth]{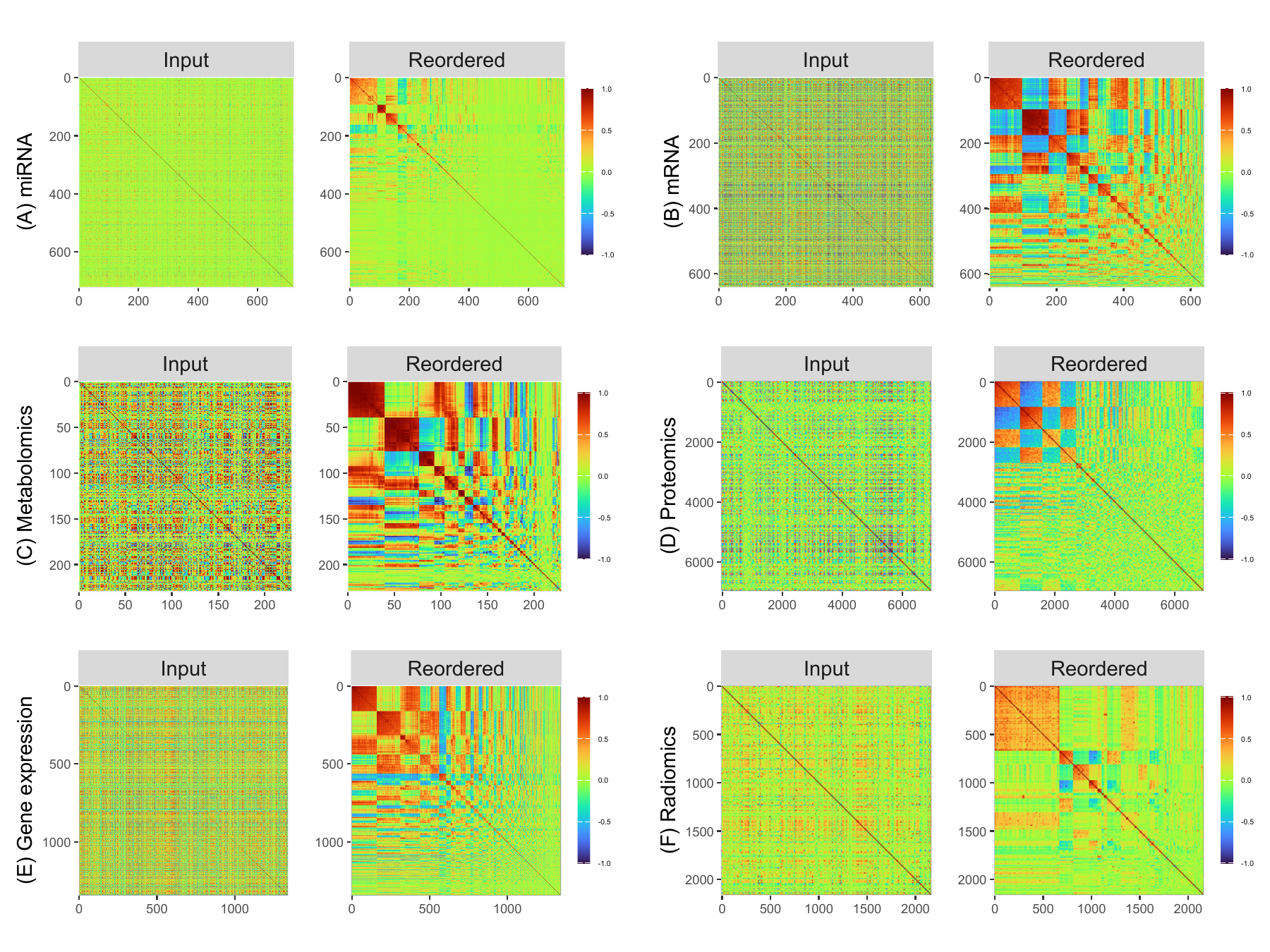}  
    \caption{Latent interconnected community/block structure in various omics data (left subfigures: input covariance/correlation matrices, right subfigures: reordered). (A) The Cancer Genome Atlas (TCGA) Pan-kidney cohort micro RNA (miRNA) data. (B) Leukocyte Messenger RNA (mRNA) data \citep{overmyer:2021large}. (C) Alzheimer's Disease Neuroimaging Initiative (ADNI) metabolomics data. (D) TCGA breast cancer proteomics data. (E) TCGA Pan-Kidney cohort Gene expression data. (F) UK Biobank Imaging-derived phenotypes (IDP) data.}
    \vspace{-5mm}
    \label{fig:omics}
\end{figure*}

\section*{Methods}
\subsection*{General Linear Multivariate Regression}

We first consider a general multivariate regression, formulated as 
\begin{align}
    \bY = \bB\bX + \bE,
    \label{eq:GLM}
\end{align}
where  $\bY=(\bY_1,\cdots,\bY_n) \in \mathbb{R}^{p\times n}$ denote the multivariate outcomes for $p$ omics features across $n$ participants, $\bX=(\bX_1,\cdots,\bX_n) \in \mathbb{R}^{q\times n}$ contains $q$ predictors (e.g., disease status) and  $\bB = (\bB_1,\cdots,\bB_p)^\top \in\mathbb{R}^{p\times q}$ is the matrix of regression coefficients, and  $\bE=(\bE_1,\cdots,\bE_n)\in\mathbb{R}^{p\times n}$ is the residual matrix. Under the general multivariate regression framework, $\bY_i$ follows a multivariate normal distribution, $\bY_i \sim \text{Normal} (\bB \bX_i, {\Cov}(\bE_i)={\bfSigma}_{p \times p})$. Model \eqref{eq:GLM} represents  traditional mass-univariate analysis methods (e.g., DE analysis), which model omics variables separately while implicitly assuming a diagonal error covariance matrix $\bfSigma$.

However, due to the (latent) structured dependence patterns commonly observed in omics data, $\Cov(\bE_i)$ deviates from a diagonal matrix (see \Cref{fig:omics}), thereby violating the assumption and leading to sub-optimal statistical inference.

\subsection*{CoReg Model Specification}\label{Sec:Method}
We propose CoReg, a model that explicitly accounts for dependence among omics outcomes in multivariate association analysis. As the dependence pattern is represented by the covariance of residuals $\Cov(\bE_i)$, CoReg focus on residual modeling to account for the network dependence patterns like other low-dimensional latent factor models (e.g., \cite{Leek(PNAS):2008}). Specifically, by modeling the residual term $\bE$ in \eqref{eq:GLM} by $\bE=\bfGamma\bF + \boldsymbol{\mathcal{E}}$,  we specify the CoReg model as follows:
\begin{align}
    \bY = \bB\bX  + \bfGamma\bF + \boldsymbol{\mathcal{E}},
    \label{eq:CoReg_model}
\end{align}
where $\bF = (\bF_1,\cdots, \bF_n) \in\mathbb{R}^{K \times n}$ denotes low-dimensional latent factors ($K < p$) that capture both dependence patterns and underlying co-expression network structure, $\bfGamma \in \mathbb{R}^{p\times K}$  is the regression coefficient matrix associated with $\bF$, and $\boldsymbol{\mathcal{E}} = (\boldsymbol{\varepsilon}_1,\cdots,\boldsymbol{\varepsilon}_n) \in \mathbb{R}^{p\times n}$ is the residual after the factor-adjustment.

Under the CoReg model in \eqref{eq:CoReg_model}, the covariance of the residual $\Cov(\bE_i)=\Cov( \bfGamma\bF_i + {\boldsymbol{\mathcal{E}}}_i)$ admits the following decomposition:
\begin{align}
    \Cov( \bfGamma\bF_i + {\boldsymbol{\mathcal{E}}}_i)=\bfGamma \Cov(\bF_i) \bfGamma^\top + \bfSigma_{\boldsymbol{\mathcal{E}}}.
    \label{eq:CoReg_model_cov}
\end{align}
This separates the covariance $\Cov(\bE_i)$ into two components: (i) a factor-driven component $\bfGamma \Cov(\bF_i) \bfGamma^\top$, which captures structured cross-omics dependence explained by the latent factors,  and (ii) a noise term $\bfSigma_{\boldsymbol{\mathcal{E}}}$ unexplained by the factors. As illustrated in \Cref{fig:Workflow_CoReg}, once $\bF$ is properly identified, the residual covariance is well approximated by the first term $\bfGamma \Cov(\bF_i) \bfGamma^\top$ (see $\Cov(\bE_i)\approx\Cov(\bL\bF_i)$ in step 2). Consequently, the factor-driven component effectively models the off-diagonal covariance, leaving a noise term $\bfSigma_{\boldsymbol{\mathcal{E}}}$ that is nearly diagonal. This indicates that the residuals are approximately independent after factor adjustment (see step 3 in \Cref{fig:Workflow_CoReg}), enabling CoReg to account for covariance and yield valid statistical inference.

CoReg consists of three main steps: \textbf{Step 1.} fit a standard multivariate regression without accounting for dependence; \textbf{Step 2.} estimate the latent factors $\bF$ that encode the outcome dependence structure, and \textbf{Step 3.} perform a factor-augmented multivariate regression analysis. A detailed workflow is provided in \Cref{fig:Workflow_CoReg} and summarized in Algorithm 1 in \tblue{SM}. Since step 2 is pivotal in CoReg, in the following, we present the detailed procedure of identifying the latent factors $\bF$ by integrating (omics) co-expression network analysis and factor model.

\subsection*{Co-expression Network Analysis}

We implement a novel co-expression network extraction procedure that is tailored to extract interconnected communities from  the correlation/co-expression matrices of omics data.  The extracted interconnected communities  explicitly reveal the block structure (see step 2 (A,B) in \Cref{fig:Workflow_CoReg}) and can be reorganized into a well-studied block correlation matrix \citep{Archakov:2024}. The extracted communities will be used to identify $\bF$ in the subsequent  analysis (see step 2(C) in \Cref{fig:Workflow_CoReg}).

The input to co-expression network analysis is the residual covariance (or correlation) matrix obtained after fitting the multivariate regression in \eqref{eq:GLM} (step 1), i.e., $\bfSigma= \Cov(\bE_i)$. When variables are standardized to unit variance, $\bfSigma$ reduces to the correlation matrix $\bR$.  We represent $\bfSigma$ as a weighted, undirected graph $G=(V,W)$, where $V=\{v_1,\ldots,v_p\}$ is the node set (node $v_j$ corresponds to omics variable $j$) and $W=\{\omega_{jj^\prime}: 1\le j<j^\prime\le p\}$ is the set of edge weights. Each weight $\omega_{jj^\prime}$ quantifies the pairwise association between omics variables $j$ and $j^\prime$ (e.g., $\omega_{jj^\prime} = \vert \texttt{cor}(j,j^\prime)\vert$ with $\bR$ as input). As the interconnected community structure is highly prevalent in various types of omics data (\Cref{fig:omics}), we consider a general graph structure $G = \bigcup_{k=1}^KG_k \cup G_0$ which decomposes collection of $K$ dense subgraphs/modules $\{G_k=(V_k,W_k)\}_{k=1,\cdots,K}$ and $G_0 = (V_0, W_0)$ comprising singletons (i.e., isolated omics variables). Subgraphs/modules can be positively and negatively correlated. 

To extract the co-expression network structure, we optimize the following objective function: 
\begin{align}
     \bigcup_{k=1}^K \widehat{G}_k = \argmax_{\{\bigcup G_k, K\}} \sum_{k=1}^K\frac{\vert W(G_k) \vert}{\vert V_k\vert^\lambda}, \label{eq:Dense}
\end{align}
where the numerator aggregates the within-module edge weights (i.e., $\vert W(G_k) \vert=\sum_{j<j^\prime \in G_k} \omega_{jj^\prime}$), and the denominator penalizes the size of the module. Thus, \eqref{eq:Dense} seeks to recover a collection of densely connected modules of minimal sizes (i.e., number of edges). The tuning parameter $\lambda \in (1,2]$ regulates the trade-off between maximizing intra-module weights and constraining module sizes. Increasing $\lambda$ places greater penalties on module sizes, leading to smaller, denser modules. The objective function \eqref{eq:Dense} can automatically selects $K$, while $\lambda$ is objectively determined (see \tblue{Algorithm 1} in \tblue{SM}).

The outputs of CoReg's co-expression network analysis includes the estimated modules $\bigcup_{k=1}^K \widehat{G}_k$ and a reordered covariance matrix $\Cov(\bE_i)$ based on these modules, which reveals an explicit interconnected community structure (see \Cref{fig:omics}). CoReg differs from commonly used co-expression network analysis toolkits by distinguishing strong correlations within diagonal blocks from medium correlations of inter-block connections. Therefore, CoReg provides a ``higher-resolution'' solution for co-expression network extraction compared to classical methods (see \tblue{Fig. S1} in \tblue{SM}), yielding a more accurate characterization of latent omics dependence patterns.

\subsection*{Co-expression Network-based Factor Model for $\Cov(\bE_i)$}

We represent the residuals using a factor model:
\begin{equation}
\begin{aligned}
    & \bE_i =\bL \bF_i + \boldsymbol{\bU_i} \\
    & \Cov(\bE_i)=\bfSigma = \bL \bfSigma_{\bF}\bL^\top + \bfSigma_{\bU}, 
\end{aligned}
\label{eq:CFA}
\end{equation}
where $\bL \in \mathbb{R}^{p\times K}$ denotes the loading matrix, $\bF_i$ the latent factors, $\bU_i$ the errors not explained by the factors, and $\bfSigma = \Cov(\bE_i)$,  $\bfSigma_{\bF} = \Cov(\bF_i)$, and $\bfSigma_{\bU}=\Cov(\bU_i)$.

Unlike commonly used exploratory factor models, which allow each omics variable to load onto multiple factors,  we implement  a non-overlapping structure on the loading matrix $\bL$ \ as in confirmatory factor analysis (CFA) models: 
\begin{align}
    \bL = \begin{pmatrix}
        {\boldsymbol{\ell}}_1 & \bf{0} & \cdots & \bf{0}\\
        \bf{0} & {\boldsymbol{\ell}}_2 & \cdots & \bf{0}\\
        \vdots & \vdots & \ddots & \vdots\\
        \bf{0} & \bf{0}  &\cdots & {\boldsymbol{\ell}}_K\\
        \bf{0} & \bf{0} &\cdots & \bf{0} 
    \end{pmatrix}, \label{eq:Loading}
\end{align}
where ${\boldsymbol{\ell}}_k \in \mathbb{R}^{p_k}$ is a loading  sub-vector that maps a set of omics variables exclusively to the $k$th factor. We specify non-zero loadings based on $\bigcup_{k=1}^K \widehat{G}_k$, with the number of factors equals to the number of  communities $K$. The nonzero elements of the $k$th column (i.e., ${\boldsymbol{\ell}}_k=(\ell_{k,1},\cdots,\ell_{k,p_k})^\top$) corresponds to the omics variables in module $G_k$ with $p_k$ nodes, while the other entries of the $k$th column are zeros. This formulation ensures that the corresponding low-rank representation $\bF$ of high-dimensional variables not only reduces dimensionality but also preserves biological interpretability, with each factor linked to an expression network module.

The factor model-based covariance matrix specification in \eqref{eq:CFA} coincides with the covariance reordered by co-expression network analysis (see step 2 (C) in \Cref{fig:Workflow_CoReg}).  Specifically, we have
\begin{equation}
\begin{aligned}
   \bfSigma &= \begin{pmatrix}
        \bfSigma_{11}& \bfSigma_{12} & \cdots & \bfSigma_{1K} & \bf{0}\\
        \bfSigma_{21} & \bfSigma_{22} & \cdots & \bfSigma_{2K} & \bf{0}\\
        \vdots & \vdots & \ddots &\vdots & \vdots\\
        \bfSigma_{K1} & \bfSigma_{K2} & \cdots & \bfSigma_{KK} & \bf{0}\\
        \bf{0} & \bf{0} & \cdots & \bf{0} & \bf{0}
    \end{pmatrix} +\bfSigma_{\bU}  \\
   &=\bL \bfSigma_{\bF} \bL^\top +\bfSigma_{\bU}\\
   &= \begin{pmatrix}
        {\boldsymbol{\ell}}_1 \bfSigma_{\bF_{(11)}} {\boldsymbol{\ell}}_1^\top &  \cdots & {\boldsymbol{\ell}}_1 \bfSigma_{\bF_{(1K)}} {\boldsymbol{\ell}}_K^\top & \bf{0}\\
        {\boldsymbol{\ell}}_2 \bfSigma_{\bF_{(21)}} {\boldsymbol{\ell}}_1^\top &  \cdots & {\boldsymbol{\ell}}_2 \bfSigma_{\bF_{(2K)}} {\boldsymbol{\ell}}_K^\top & \bf{0}\\
        \vdots &  \ddots & \vdots &\vdots\\
        {\boldsymbol{\ell}}_K \bfSigma_{\bF_{(K1)}} {\boldsymbol{\ell}}_1^\top &  \cdots & {\boldsymbol{\ell}}_K \bfSigma_{\bF_{(KK)}} {\boldsymbol{\ell}}_K^\top & \bf{0}\\
         \bf{0} & \cdots & \bf{0} & \bf{0}
    \end{pmatrix}  +\bfSigma_{\bU}, 
\end{aligned}
    \label{eq:UB_Matrix}
\end{equation}
where the first line represents the block covariance matrix of omics data, and each submatrix $\bfSigma_{kk^\prime}$ represents the intra-community correlation strengths when $k = k^\prime$ (diagonal blocks) and inter-community correlation strengths when $k \neq k^\prime$ (off-diagonal blocks); the second line is the CFA covariance model in \eqref{eq:CFA}; and the third line expresses each  omics covariance matrix block in terms of  corresponding loading vectors and factor covariance matrix. The above links the block/community memberships in co-expression network analysis with loadings in CFA based on the shared structured covariance matrix (see formal derivation in \tblue{proposition 1} in \tblue{SM}). When the dependence pattern among omics variables is well captured by the interconnected community structure, the covariance matrix $\bfSigma$ can be largely explained by the  first term (i.e., loadings $\bL$ (i.e., module memberships)  and covariance of factors  $\bfSigma_{\bF}$) leaving $\bfSigma_{\bU}$ approximately equal to the diagonal matrix (i.e., independence). Therefore, CoReg accounts for the covariance among omics variables by introducing a low-dimensional vector $\bF$ into the multivariate regression model.

We follow the computational procedure of \citet{yang:2024semi} to estimate $ {\boldsymbol{\ell}}$ and $\bF$ using closed-form estimators. CoReg's factor estimation offers two key benefits: (i) by allowing correlated factors, it captures the interconnected modular structure of omics data and better characterizes dependence than conventional latent-factor models, which commonly enforce factor orthogonality \citep{Leek(PNAS):2008}; and (ii) it achieves faster estimation via co-expression–guided loadings, which avoid the computationally expensive iterative matrix decompositions required by standard approaches \citep{Dhaene:2024}.

\begin{figure*}[t!]
\centering
\resizebox{0.8\textwidth}{!}{ 
\begin{tikzpicture}[node distance=1cm, font=\fontsize{8}{11}\selectfont]
    \node (box1) [box, header = Step 1: Multivariate regression ]{\parbox{8cm}{\centering \vspace{0.5cm} 
    \begin{minipage}{1\linewidth}
    \centering
                \begin{itemize}
                    \item[] $\bY_{i} = \bB \bX_i +$ \colorbox{yellow}{$\bE_{i}$}
                \end{itemize}
    \end{minipage}\\\vspace{.3cm}
    \begin{minipage}{1\linewidth}
            \begin{tikzpicture}[node distance=0.01cm]
            \centering
            \node(Blank1) {};        
            \node(S1) [left = -1cm of Blank1]{
            \scalebox{0.4}{\includegraphics{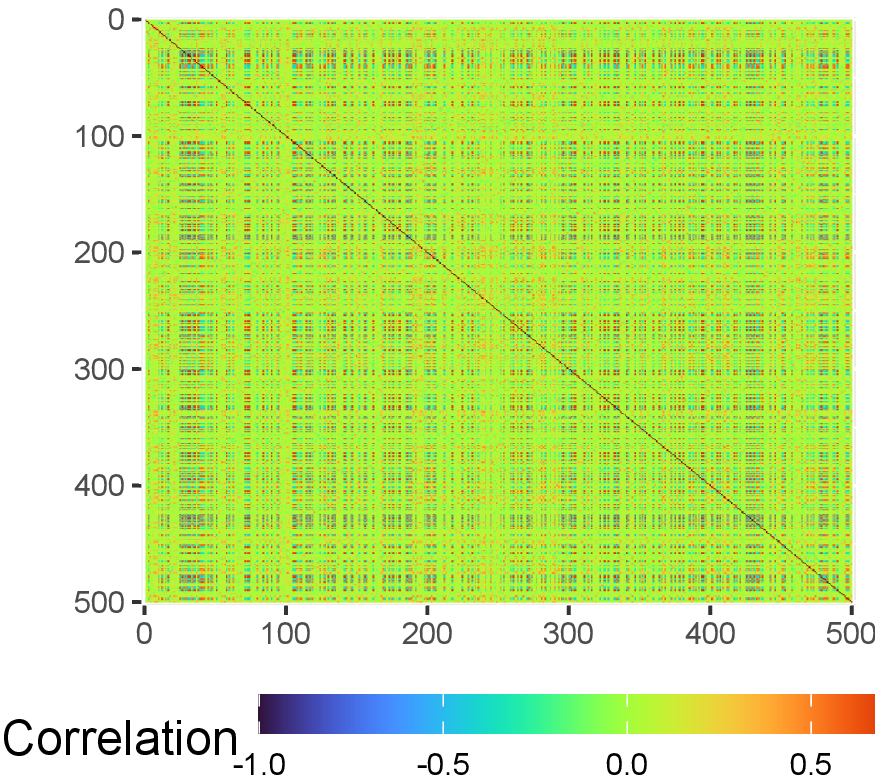}}};
            \end{tikzpicture}
    \end{minipage}}};
\node (box3) [box, header = {\shortstack{Step 3: Multivariate regression \\ with co-expression-modeled dependence}}, right=of box1]{\parbox{8cm}{\centering \vspace{0.5cm} \begin{minipage}{1\linewidth}
            \centering
                \begin{itemize}
                    \item[] $\bY_{i} = \bB \bX_i + \bfGamma \bF_i +$ \colorbox{pink}{$\boldsymbol{\mathcal{E}}_{i}$}
                \end{itemize}
            \end{minipage}\\\vspace{.3cm}
            \begin{minipage}{1\linewidth}
            \begin{tikzpicture}[node distance=0.01cm]
            \centering
            \node(Blank3) {};
            \node(S3) [right = -2.5cm of Blank3]{\scalebox{0.4}{\includegraphics{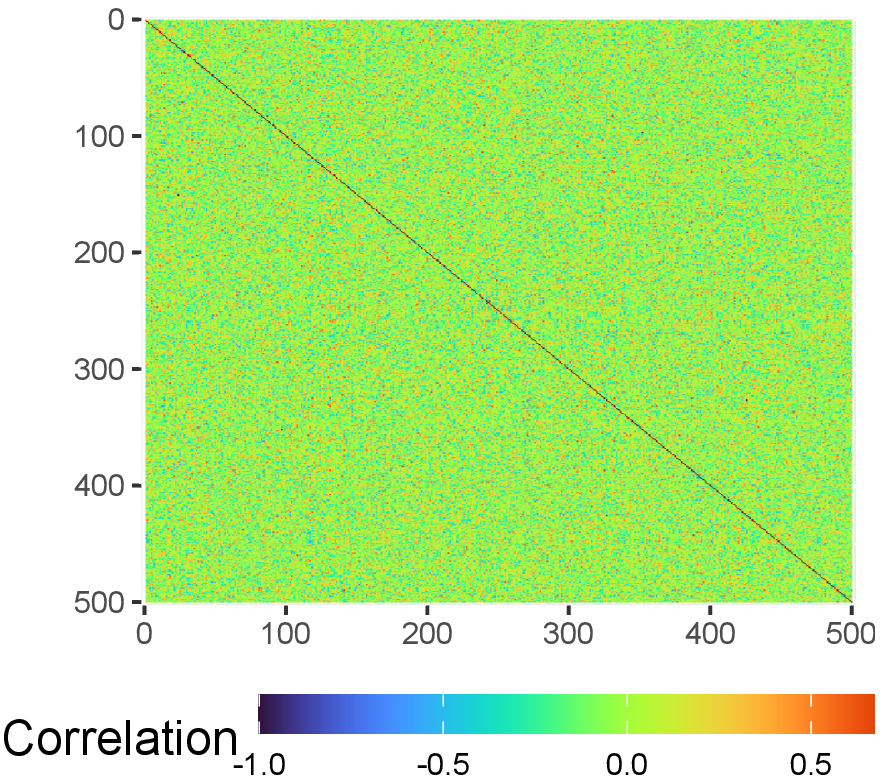}}};
            \end{tikzpicture}
            \end{minipage}}};
\path (box1.west) -- (box3.east) coordinate[midway] (midpoint);

   \node (box2) [box, header = Step 2: Model dependence by Co-expression network-based factor analysis, below= 4cm of midpoint] {\parbox{17cm}{ \centering  
    \begin{minipage}{1\linewidth}
            \begin{tikzpicture}[node distance=0.01cm]
                \node(Blank) {};
                \node(S2_Unordered) [left = -1cm, below=-1.4cm of Blank]{\scalebox{0.4}{\includegraphics{Workflow1_Cor.eps}}};
                \node(S2_Ordered) [right = 0.8cm of S2_Unordered]{\scalebox{0.4}{\includegraphics{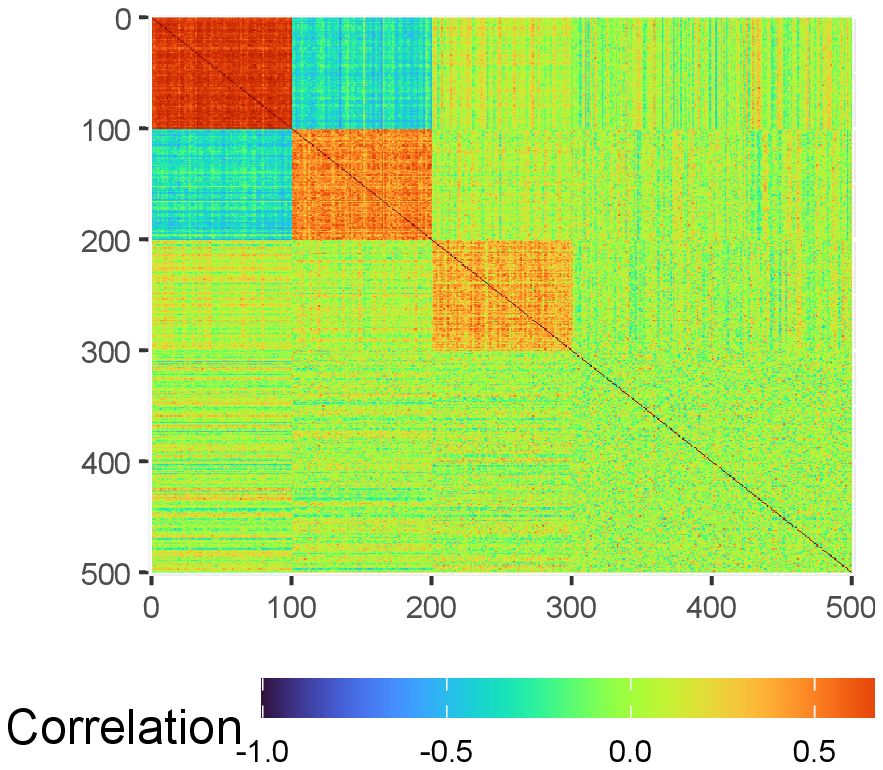}}};
                \node(S2_Clean) [right = 0.8cm of S2_Ordered]{\scalebox{0.4}{\includegraphics{Workflow3_Clean.eps}}};
                
                \node(Sig) [below = 1cm of S2_Unordered]{\scalebox{0.4}{\includegraphics{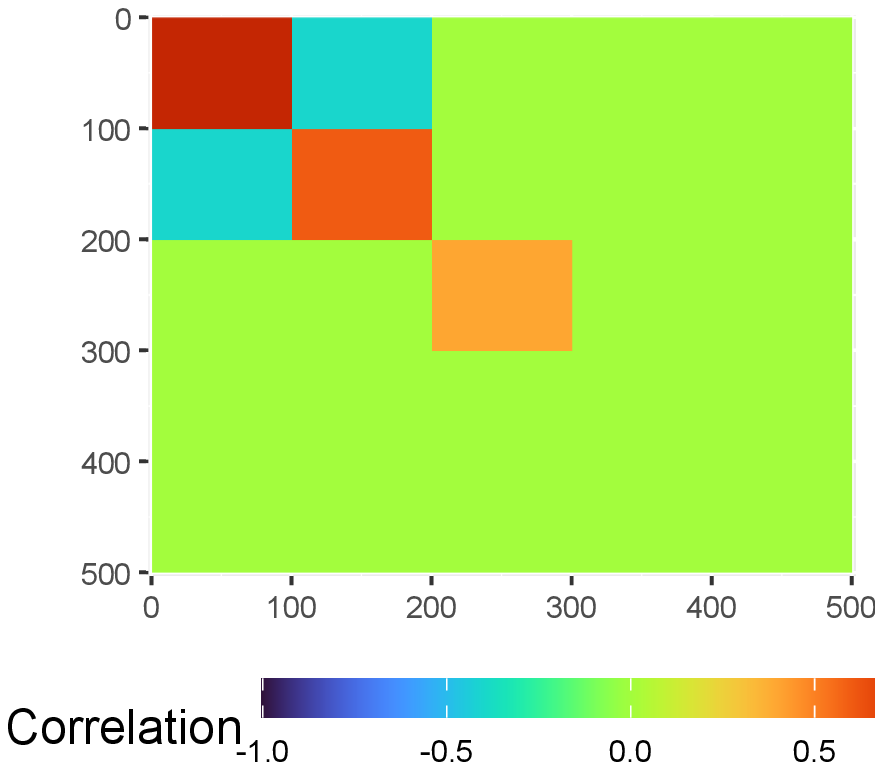}}};
                \node(Eq) [right = -2.5cm of Sig] {\Large{$=$}};
                \node(L) [right = -1cm of Sig]{\scalebox{0.4}{\includegraphics{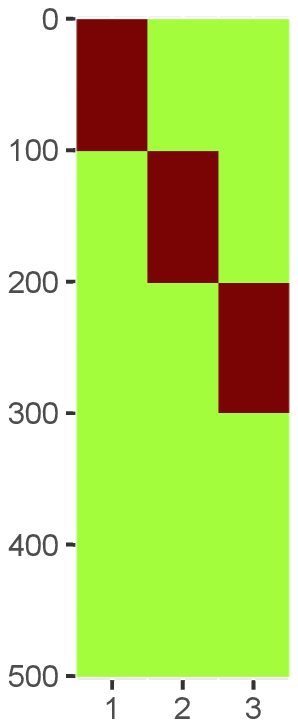}}};
                \node(F) [right = -2cm of L] {\scalebox{0.4}{\includegraphics{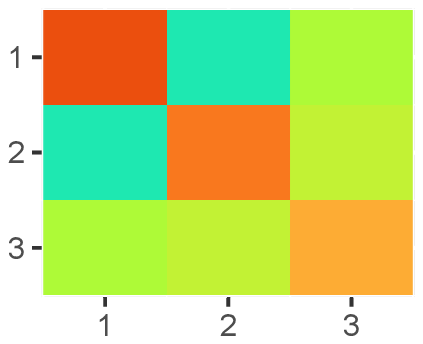}}};
                \node(Lt) [right = -1cm of F] {\scalebox{0.4}{\includegraphics{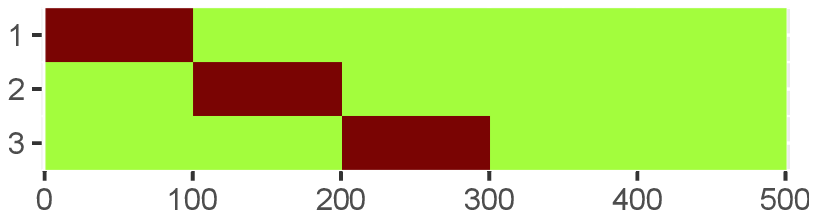}}};
                \node(Prod1) [right = -3.5cm of L] {\Large{$\times$}};
                \node(Prod2) [right = -3.5cm of F] {\Large{$\times$}};

                \node(L_some) [left = -1.8cm of L.east]{};    
                \node(L_text) [below = -1.1cm of L]{$\bL$};
                \node(F_text) [right = -1.8cm of L_text]{$\Cov(\bF_i)$};
                \node(Lt_text) [right = -1cm of F_text]{$\bL^\top$};

                \node(SCFA_above) at ($(S2_Ordered)+(0.5cm,-3cm)$) {};
                
                \draw[arrow, dashed, line width=0.1mm] ($(S2_Ordered.south)+(0.5cm,0)$) -- (SCFA_above) node[midway, right, xshift=- .2cm] {\colorbox{myblue}{\small{\color{white} Confirmatory Factor Analysis}}};

                \draw[arrow, dashed, line width=0.1mm] ($(S2_Unordered.east)+(-0.5cm,0)$) -- ($(S2_Ordered.west)+(0.5cm,0)$)node(Coexp)[midway, above, yshift=-0.5cm, align=center] {\colorbox{mypurple}{\small{\color{white}Coexpression}}};
               \node[align=center](coexp) at ($(Coexp)+(0.1cm,+0.2cm)$) {{\tiny{\color{mypurple}\textbf{\shortstack{Reorder by \\the extracted modules}}}}};
                \draw[arrow, line width=0.5mm,color=myred] ($(S2_Clean.west)-(-0.5cm,0)$) -- ($(S2_Ordered.east)+(-0.5cm,0)$)node[midway, above, yshift=-0.5cm] {\color{myred}{\Large{$\boldsymbol{+}$}}};
                \draw[arrow, line width=0.5mm,color=myred] ($(Sig.east)+(-0.8cm,1.5cm)$) -- ($(S2_Ordered.west)+(0.5cm,-1cm)$) node[midway, above,xshift=-0.4cm, yshift=-0.9cm] {\color{myred}{\Large{$\boldsymbol{+}$}}};
                \node[draw, rectangle, inner xsep=-1.6cm, inner ysep=0.5cm, fit={(L_some) (F) (Lt)  ($(Lt.east)+(1.5cm,0)$) ($(Lt.east)+(-12cm,0)$) ($(L_text.south)+(0,1cm)$)}] (box.scfa) {};     
            \end{tikzpicture}
            \end{minipage}\\
                \begin{minipage}{1\linewidth}
            \begin{itemize}
               \item[] \hspace{0.5 cm} Model dependence by a confirmatory factor model:   $\bE_{i} = \bL \bF_i + \bU_{i}$ $ \Leftrightarrow \Cov(\bE_i)=\bL \Cov(\bF_i)\bL^\top + \Cov(\bU_i)$\vspace{.1cm}
               \item[] \hspace{0.5 cm} Estimate factors $\bF$ from co-expression networks:   $\bF_i = {({\widehat{\bL}}^\top \widehat{\bL})}^{-1} \widehat{\bL}^\top \bE_i$, where loadings $\widehat{\bL}$ represent modules in  \eqref{eq:Loading}
            \end{itemize}
            \end{minipage}}
            };
       
\draw [arrow, line width=0.2mm] (box1.west) -- ++(-0.5cm, 0) |-  (box2.west);
\draw [arrow, line width=0.2mm] (box2.east) -- ++(0.5cm, 0) |-  (box3.east);

\node (S1_FN) at ($(S1.south)+(2.8cm,-0.3cm)$) {};
\node (S1_FS) at ($(S1.south)+(2.8cm,-3.3cm)$) {};
\draw[arrow, dashed, line width=0.1mm] (S1_FN) -- (S1_FS) node[] {};   

\node (S3_FN) at ($(S3.south)+(5.5cm,-0.3cm)$) {};
\node (S3_FS) at ($(S3.south)+(5.5cm,-4.2cm)$) {};

\node (S1_caption) at ($(S1.north)+(6.5cm,-2.3cm)$) {\colorbox{yellow}{$\Cov(\bE_i)$}};
\node (S1_caption_end) at ($(S1_caption.east)+(-3.2cm,0cm)$) {};
\draw[arrow, dashed, line width=0.1mm] (S1_caption) -- (S1_caption_end) node[] {};   

\node (S2_caption1) at ($(S2_Ordered.south)+(-1.5cm,-0.5cm)$) {(B) $\Cov(\bE_i)_\text{Reordered}$};
\node (S2_caption2) at ($(Sig.south)+(-2cm,-0.5cm)$) {(C) $\Cov(\bL\bF_i)$};
\node (S3_caption) at ($(S3.north)+(1.3cm,-2.3cm)$) {\colorbox{pink}{$\Cov(\boldsymbol{\mathcal{E}}_i)$}};
\node (S3_caption_end) at ($(S3_caption.west)+(3.7cm,0cm)$) {};
\draw[arrow, dashed, line width=0.1mm] (S3_caption) -- (S3_caption_end) node[] {};  
\node (S2_captionA) at ($(S2_caption1)+(-6.5cm,0cm)$) {(A) $\Cov(\bE_i)$};
\node (S2_captionD) at ($(S2_caption1)+(5.2cm,0cm)$) {(D) $\Cov(\bU_i)$};
\end{tikzpicture}
}
\caption{Overview of CoReg. \textbf{Step 1:} Perform traditional multivariate regression \eqref{eq:GLM}, and consequently $\Cov(\bE_i)$ characterizes the dependence patterns. \textbf{Step 2:} Model the dependence matrix $\Cov(\bE_i)$ by synchronized co-expression network analysis and confirmatory factor analysis (CFA). (A) visualizes the heatmap of correlation matrix $\Cov(\bE_i)$; (B) demonstrates the co-expression network analysis results by reordering the variables to highlight extracted modules; (C) shows that the reordered covariance matrix $\Cov(\bE_i)$ in (B) can be largely approximated by the CFA covariance model; and (D) displays $\Cov(\bU_i)$ is approximately a diagonal matrix as $\Cov(\bF_i)$ can well represent the dependence patterns. \textbf{Step 3:} Perform multivariate analysis using $\bX$ and $\bF$ as predictors, followed by statistical inference.}
\vspace{-5mm}
\label{fig:Workflow_CoReg}
\end{figure*}

\subsection*{Regression parameter estimation and inference}   
In step 3, we fit a factor-augmented multivariate regression, including the estimated latent factors $\bF$ from step 2 as fixed predictors together with $\bX$.

The joint estimation of regression coefficients $(\bB,\bfGamma)$ can be obtained by least squares method, i.e., $(\widehat{\bB},\widehat{\bfGamma}) = \argmin_{\bB\in\mathbb{R}^{p\times q},\bfGamma\in\mathbb{R}^{p\times K}} \Vert \bY- \bB\bX - \bfGamma\bF \Vert_F^2$. By the orthogonality of $\bX$, and $\bF$ (see \tblue{Lemma 1} in \tblue{SM}), including $\bF$ does not affect the estimation of $\bB$, resulting in $\widehat{\bB}$ that is identical to that estimated from model in \eqref{eq:GLM} (see \tblue{Lemma 2} in \tblue{SM}). Estimation of $\bfGamma$ also does not depend on the presence of $\bX$ (see \tblue{Lemma 2} in \tblue{SM}). Moreover, the joint estimation of $(\bB,\bfGamma)$ is equivalent to sequentially estimating each one - first, regressing $\bY$ on $\bX$ and then regressing the residuals (i.e., $\bE=\bY-\widehat{\bB}\bX$) on $\bF$ (i.e., $\argmin_{\bfGamma}\Vert \bE- \bGamma\bF \Vert_F^2$, see  \tblue{Lemma 3} in \tblue{SM}).

CoReg accounts for the dependence among omics outcomes and ensures that the residuals $\boldsymbol{\mathcal{E}} = \bY -\bB\bX - \bfGamma\bF$ being independent (see  \tblue{Lemma 4} in \tblue{SM}). We further show $\Vert \bE - \bfGamma\bF \Vert_F^2 \leq \Vert \bE-\bL\bF\Vert_F^2$, which implies $\tr(\bfSigma_{\boldsymbol{\mathcal{E}}}) \leq \tr(\bfSigma_{\bU})$ (see  \tblue{Lemma 5} in \tblue{SM}). Thus, rather than directly subtracting $\bL\bF$, incorporating $\bF$ into the regression model as fixed predictors and allowing $\bfGamma$ to be freely estimated improves model fit. This reduces residual variance, resulting in smaller standard errors in test statistics and improved sensitivity in identifying true associations.

\textit{Inference.} To assess the association between the $l$th omics outcome and the $m$th predictor through hypothesis testing, i.e.,  $H_0: \bB_{lm}=0$ vs $H_1: \bB_{lm}\neq0$, the test statistic can be defined by $t_{lm} = \widehat{\bB}_{lm}/se(\widehat{\bB}_{lm}),$ where $se(\widehat{\bB}_{lm})=\sqrt{\widehat{\bfSigma}_{\boldsymbol{\mathcal{E}}_{ll}}[(\bX\bX^\top)^{-1}]_{mm}}$, and $t_{lm}$ follows a $t$-distribution with the degree of the freedom $n-(q+K)$.
We next perform the multiple testing correction. The independence among test statistics is a key assumption of many multiple testing correction methods, for example, the commonly used FDR control \citep{BHFDR:1995}. In CoReg, the near-diagonal structure of $\Cov(\boldsymbol{\mathcal{E}}_i)$ reflects the approximate independence of residuals, ensuring that the multiple tests of omics variables can be treated as independent. Thus, applying  valid multiple-testing correction procedure with CoReg improves the reliability and replicability of simultaneous high-dimensional inference. 

\section*{Simulation Study}\label{Sec:Simulation}

We conducted simulation studies to assess the performance of CoReg under various settings. We also evaluated the replicability of findings by CoReg between two independent simulated datasets with identical true signals and benchmarked it with existing methods.

We simulated high-dimensional correlated omics variables $\bY$ from a multivariate normal distribution: $\bY_i \sim \text{MVN} (\mathbf{B} \bX_i, \bSig)$, where we set the dimension of the response variable as $p=500$, $\bX_i$ as fixed predictors, and the corresponding coefficients $\mathbf{B}$. We specified the large covariance matrix  $\bfSigma$ based on an interconnected block diagonal form, reflecting the co-expression network structure commonly encountered in omics research. Specifically, we constructed the covariance matrix $\bfSigma$, with three dense diagonal blocks, each comprising 100 omics variables. The intra-correlations within these blocks varied: the first block exhibited a mean correlation of 0.8, the second block had a mean correlation of 0.6, and the third had a mean correlation of 0.4. Within each block, correlation values were generated from a normal distribution centered at the respective mean correlation value for that block. Furthermore, we assigned a blockwise negative inter-correlation of -0.4 between the first and second blocks. The heatmap of the correlation structure used in this simulation study is shown in \tblue{Fig. S3} in  \tblue{SM}.  We considered several settings with various sample sizes $n=\{200,500\}$ and noise levels $\sigma^2=\{0.5,1\}$. For each setting, we repeated the above simulation procedure 100 times.

\subsection*{Statistical Inference Accuracy} We applied CoReg, Ordinary Least Squares (OLS), and SVA to the simulated datasets to identify outcomes that are associated with predictors. For each method, omics variables were selected if corresponding corrected $p$-values of $\widehat{\mathbf{B}}$ fell below the specified threshold. The performance of each method was evaluated by comparing selected suprathreshold omics variables, with reference to the ground truth of those $\mathbf{B}\neq 0$. The assessment metrics include sensitivity, specificity, $F_1$ score, the area under the ROC curve (AUC), and the false discovery rate (FDR), as summarized in \Cref{Tab:Simulation1}. ROC curves, derived from sensitivity and specificity across varying cut-off thresholds applied to the regression $p$-values of all omics variables, are shown in \Cref{fig:ROC}, where a higher AUC reflects better recovery of omics variables truly associated with the covariate while controlling false positive findings. CoReg outperforms the competing methods across three key metrics-sensitivity, $F_1$ score, and AUC. Although OLS yields higher specificity and lower FDR, these values are largely driven by its extremely low sensitivity; by classifying most signals as negatives, OLS attains these metrics at the cost of a high false negative rate. In contrast, CoReg can maximally identify true signals while controlling false positive findings, as evidenced by its high $F_1$ score, which reflects a balanced trade-off between sensitivity and false discovery, thus offering more reliable findings.

\begin{table*}[h]
\centering
\caption{Simulation results: the sensitivity, specificity, $F_1$, AUC, and FDR were compared for all methods across 100 replications.}
\label{Tab:Simulation1}
\begin{tabular}{llrrrrrrrr}
\toprule
$\sigma^2$& & \multicolumn{3}{c}{$n=200$} & & \multicolumn{3}{c}{$n=500$} \vspace{.2cm} \\
& & OLS & SVA & CoReg  & & OLS & SVA & CoReg\\ 
  \cmidrule{3-5} \cmidrule{7-9}
\multirow{5}{*}{$0.5$} &  Sensitivity &0.214 (0.21) & 0.288 (0.15) & \textbf{0.664} (0.15) &&  0.823 (0.12) & 0.848 (0.15) & \textbf{0.948} (0.05) \\ 
&  Specificity &  \textbf{0.994} (0.01) & 0.956 (0.02) & 0.963 (0.02)  &&  \textbf{0.977} (0.01) & 0.949 (0.02) & 0.950 (0.02) \\ 
&  $F_1$ &  0.305 (0.27) & 0.417 (0.18) & \textbf{0.776} (0.11)   &&  0.889 (0.08) & 0.893 (0.10) & \textbf{0.955} (0.03) \\ 
&  AUC & 0.852 (0.08) & 0.683 (0.09) &\textbf{ 0.905} (0.06) && 0.975 (0.03) & 0.938 (0.07) & \textbf{0.986} (0.02) \\ 
 &FDR & \textbf{0.018} (0.03) & 0.159 (0.18) & 0.041 (0.03) &&  \textbf{0.020} (0.01) & 0.045 (0.03) & 0.037 (0.02) \\ 
   \hline
\multirow{5}{*}{$1$} & Sensitivity & 0.030 (0.07) & 0.061 (0.11) & \textbf{0.407} (0.17)  &&  0.346 (0.21) & 0.414 (0.15) & \textbf{0.753} (0.12) \\ 
&  Specificity &  \textbf{0.998} (0.00) & 0.986 (0.01) & 0.972 (0.02) && \textbf{0.988} (0.01) & 0.933 (0.02) & 0.952 (0.02) \\ 
&  $F_1$ &0.051 (0.11) & 0.095 (0.17) & \textbf{0.550} (0.17) && 0.473 (0.24) & 0.552 (0.15) & \textbf{0.837} (0.09) \\ 
&  AUC & 0.731 (0.09) & 0.569 (0.07) & \textbf{0.814} (0.08)  &&  0.888 (0.06) & 0.737 (0.09) & \textbf{0.932} (0.05) \\ 
&FDR &  0.079 (0.19) & 0.604 (0.42) & \textbf{0.051} (0.05)   && \textbf{0.030} (0.06) & 0.127 (0.10) & 0.045 (0.03) \\ 
\bottomrule
\end{tabular}
\end{table*}

\begin{figure}[h]
    \centering
    \includegraphics[width=0.44\linewidth]{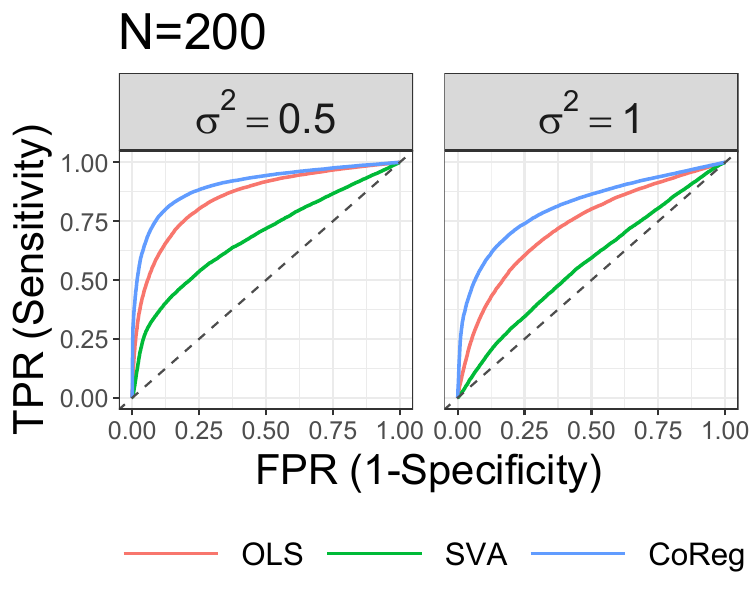}
    \includegraphics[width=0.44\linewidth]{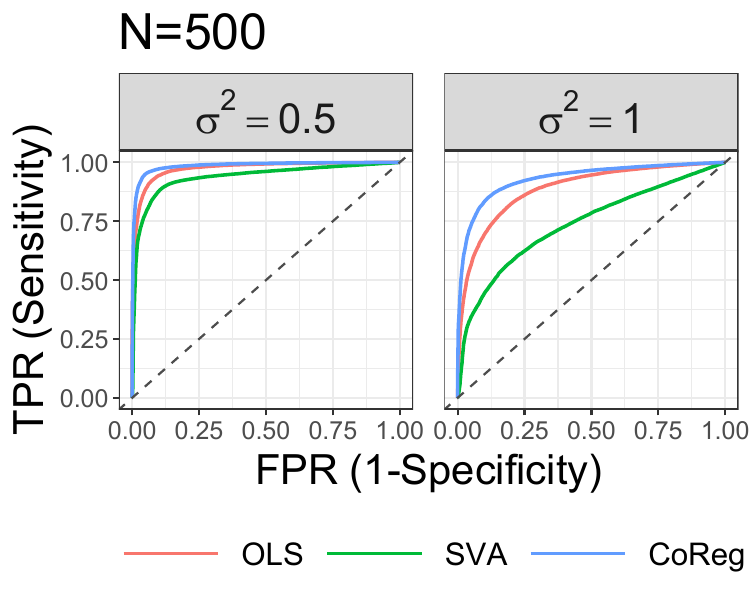}
    \caption{ROC curves (Top: $n=200$, Bottom: $n=500$) illustrating the performance of different methods in selecting true phenotype-related omics variables under varying cut-off thresholds.}
    \label{fig:ROC}
        \vspace{-5mm}
\end{figure}

Beyond inference accuracy, we evaluated computational efficiency. CoReg scales well to large sample sizes ($n$) and high-dimensional outcomes ($p$), showing only minimal runtime increases under conditions that burden traditional methods (see \tblue{Fig. S2} in \tblue{SM}).

\subsection*{Replicability Analysis}  To assess whether CoReg's improved inference accuracy also translates to enhanced cross-study replicability, we conducted an additional simulation analysis using synthetic datasets designed to emulate experiments conducted by two laboratories. To this end, we generated two independent datasets following the same data generation procedure in the previous simulation study. The two datasets shared identical true signal locations (i.e., the same non-zero entries in $\bB$), while other conditions, including effect sizes and noise levels, were varied to reflect the heterogeneity of experimental conditions. Specifically, we considered two scenarios. In the first scenario, we varied the noise levels between datasets-setting $\sigma^2=0.5$ for dataset 1 and $\sigma^2=1$ for dataset 2-while keeping the effect size at $0.3$. In the second scenario, we varied the effect sizes-assigning $0.3$ for dataset 1 and $0.1$ for dataset 2-while keeping the noise level at $\sigma^2 = 0.5$. For each scenario, we used a sample size of $n=200$ and conducted 100 replications. True positives were counted for each dataset, and replicability performance was evaluated based on the number of true positives shared across the two datasets.

The results from the first scenario are presented as a Venn diagram in \Cref{fig:Replicability} (see results of scenario 2 in \tblue{SM}). A larger intersection area represents higher replicability, suggesting that the method yields consistent discoveries across different datasets. CoReg's improved sensitivity (larger number of true positives in both datasets) leads to higher replicability (larger overlap in discoveries). Specifically, the intersection area (number of overlapped true positive findings from dataset 1 and dataset 2) for CoReg is 195 (64.8\%), compared to OLS at 63 (26.4\%) and SVA at 87 (35.4\%).

\begin{figure*}[h]
    \centering
\includegraphics[width=0.75\linewidth]{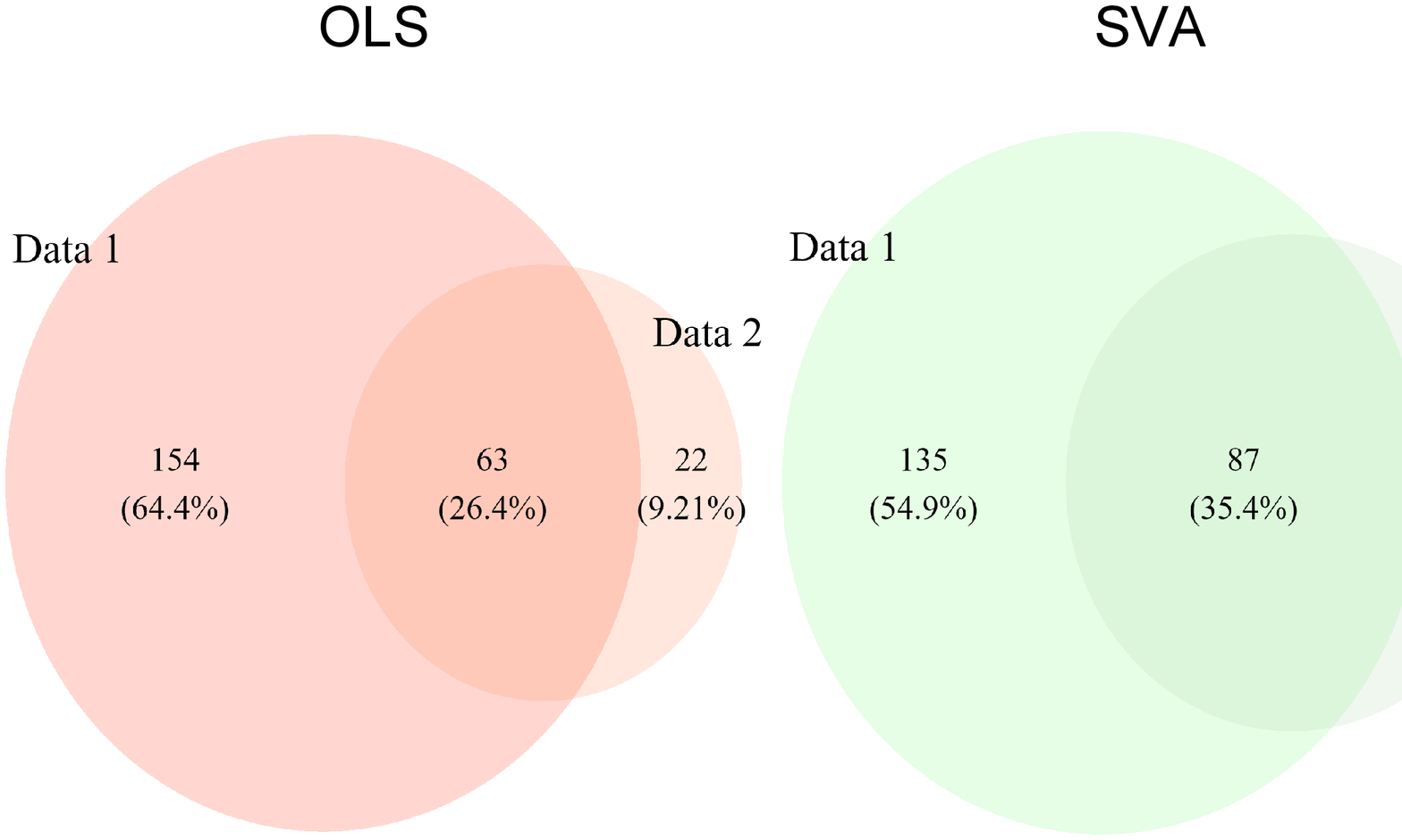}
    \caption{Replicability analysis results. Two omics studies with identical covariance structure and predictor-associated variables but differing noise levels were simulated. OLS, SVA, and CoReg were applied separately, and replicability was assessed by the overlap of true positives. CoReg yielded substantially higher overlap by recovering more true signals in both datasets.}
    \label{fig:Replicability}
    \vspace{-5mm}
\end{figure*}
\vspace{-5mm}

\section*{Age-related Omics Data Analysis}\label{Sec:Real}

CoReg is generally applicable to high-dimensional continuous omics outcomes. Here, we demonstrate the application of CoReg to two high-dimensional omics datasets. The first is a DNA methylation-based epigenetics dataset from the Alzheimer's Disease Neuroimaging Initiative (ADNI; \cite{Vasanthakumar:2020}), and the second is a nuclear magnetic resonance (NMR)-based metabolomics dataset from the UK Biobank (UKBB, \cite{Julkunen:2023}).

\subsection*{ADNI DNA Methylation data} We focus on DNA methylation measures at CpG (cytosine-phosphate-guanine) sites ($p = 4,200$) from $n = 366$ participants from the ADNI data study. We consider age and Alzheimer's disease diagnosis (AD) as the primary predictors of interest while adjusting covariates including biological sex. DNA methylation—the process by which methyl groups are added to the DNA molecule—is one of the key epigenetic mechanisms that regulate gene expression, significantly influencing development, aging, and overall health. As people age, DNA methylation levels tend to show a global decrease, reflecting significant epigenetic changes that influence cellular function and the aging process \citep{Horvath:2013}. Moreover, numerous studies have shown that DNA methylation patterns are altered in the Alzheimer's disease group, suggesting that these epigenetic alterations may be involved in the disease's pathogenesis \citep{MASTROENI:2010}. We applied CoReg to this dataset to identify CpG sites associated with age and AD diagnosis, and compared its performance with competing methods.

CoReg detected 543 CpG sites significantly associated with age (12.93\%) after multiple-testing correction, compared with only 32 (0.76\%) by OLS and 10 (0.24\%) by SVA. For AD diagnosis, CoReg identified 54 CpG sites (1.29\%) with significant associations, whereas neither OLS nor SVA detected any. These findings are well aligned with prior studies; for example, \citep{Sun:2023} reported 52 CpG sites integratively associated with DNA methylation, gene expression, and AD. In addition, CoReg identified \texttt{cg00928580}, annotated to the clusterin (CLU) gene on chromosome 8, a well-established Alzheimer’s disease risk locus with multiple prior studies reporting AD-associated differential expression of CLU \citep{Foster:2019(CLU)}. CoReg also detected \texttt{cg17132004}, annotated to SPI1 gene (also known as PU.1), a transcription factor that are linked to AD \citep{Shireby:2022}, particularly through its regulatory role in microglial function \citep{PIMENOVA:2021}. In light previous findings of strong associations between DNA methylations and both age or AD, these results highlight the improved sensitivity of CoReg, whereas OLS and SVA suffer from a lack of rejection power in identifying meaningful signals. In \tblue{Fig. S7} in \tblue{SM}, we systematically investigated the distributions of $p$-values of all CpG sites for methods, and found that the distribution of CoReg facilitates easier separation of significant findings  potentially driven by increased rejection power. The improved sensitivity is well aligned with the findings in the simulation study.

\subsection*{UKBB Metabolomics data} 
In the second application, we evaluated associations between 249 plasma metabolites measured by NMR spectroscopy and liver health and age for $n=7328$ participants in the UKBB.
It has been well documented that blood metabolites change with age \citep{Yu:2012MetaAge}. Furthermore, considering the liver's central role in metabolism, we investigated the association between metabolic profiles and liver condition, as assessed by FR-PDFF (Fat-Referenced Liver Proton Density Fat Fraction), a noninvasive magnetic resonance imaging (MRI)-derived metric that quantifies liver fat content.

Similar to the ANDI study and simulation study, our sensitivity is much improved (see  \tblue{Table S1} in \tblue{SM}). Specifically, for age, CoReg detected 245 associations (98.39\%) compared to 229 (91.97\%) identified by both OLS and SVA. For FR-PDFF, CoReg identified 240 associations (96.39\%) versus 199 (79.92\%) by OLS and SVA. This strong association obtained by CoReg is consistent with previous research \citep{Gnatiuc:2023LiverMeta}, which reported that 152 out of 180  (84.44\%) metabolites were significantly associated with PDFF at a stringent threshold of $p < 0.0001$. Additionally, CoReg completed the analysis in 0.43 seconds, whereas SVA required 3724.31 seconds, indicating that CoReg is more computationally efficient.

We conducted additional validation experiments by randomly sub-sampling the original data into subsets with varying sample sizes ranging from $n= 500$ to 5000 (in increments of 500), performing 20 repetitions at each size to assess the performance (e.g., statistical power) of all methods across different sample sizes. Statistical theory establishes that larger sample sizes generally yield greater power and, consequently, higher sensitivity across omics variables. Accordingly, a more powerful method can detect a number of significant findings comparable to those from the full dataset, even with a smaller sample size. As demonstrated in \tblue{Fig. S8} in \tblue{SM}, CoReg can identify a number of significant findings nearly comparable to those from the full dataset even with small samples (e.g., $n=500$), whereas competing methods recover only about one-third as many. These results suggest that CoReg achieves substantially greater statistical power for datasets with relatively small sample sizes (e.g., in the hundreds). Given that the sample size of most omics studies is below this range, CoReg is particularly advantageous in improving sensitivity and thus cross-study replicability.

\section*{Discussion}\label{Sec:Conclusion}

We proposed, CoReg, a new multivariate statistical framework for the analysis of omics data which accounts for dependence structures. 
CoReg extracts latent interconnected co-expression communities and translates them into factors that maximally explain the residual covariance matrix. Compared to classical methods such as DE tools, CoReg  improves the accuracy of simultaneous multivariate inference and consequently  enhances replicability across studies. In the diverse omics datasets we examined, all covariance matrices consistently exhibited an (latent) interconnected community structure. For omics covariance matrices that lack an interconnected community or block structure, CoReg performs comparably to the mass-univariate approach—analogous to the well-established robustness of GEE under misspecified working correlation structures. In addition,  we implemented computationally efficient algorithms for CoReg, with an average runtime of a few minutes for hundreds of thousands of correlated outcomes. In conclusion, CoReg provides an alternative to existing omics analysis tools by explicitly modeling complex dependencies among high-dimensional variables and is potentially applicable to a wide range of omics data. 

CoReg is publicly available at \url{https://github.com/hwiyoungstat/CoReg}.

\bibliographystyle{apalike} 
\bibliography{reference_etal}

\newpage
{\LARGE{\textbf{Supplementary Material (SM) for ``Modeling Dependence in Omics Association Analysis via Structured Co-Expression Networks to Improve Power and Replicability''}}}

\setcounter{figure}{0}
\renewcommand{\thefigure}{S\arabic{figure}}

\setcounter{table}{0}
\renewcommand{\thetable}{S\arabic{table}}


\section{Supporting Propositions and Lemmas for CoReg}

\begin{proposition}\label{prop1}
Assume that the pairwise correlations among variables depend only on their community, i.e., $\Cor(\bY^{j},\bY^{j^\prime}) = \rho_{kk^\prime}$ when $j \in G_k$, and $j^\prime \in G_{k^\prime}$. If the nonzero elements in the $k$th column of the loading matrix $\bL$ match the omics variables assigned to community $G_k$ (i.e., $\mathrm{Supp}(\bL_{\cdot k}) = G_k$ for all $k=1,\cdots,K$), then the CFA-derived covariance matrix $\bL \bfSigma_{\bF}\bL^\top$ recovers the omics correlation matrix.
\end{proposition}
Given the $K$ communities $\{G_1,\cdots,G_K\}$ identified by the dense subgraph extraction, we express the loading matrix $\bL$ such that its $(j,k)$th entry is given by $$\bL_{jk} = \begin{cases}
    1, & \text{if} \ j \in G_{k}\\
    0, & \text{otherwise}
\end{cases}, j \in 1,\cdots, p, \ k=1,\cdots,K.$$ We further define the $K \times K$ factor to factor covariance matrix $\bfSigma_{\bF}(k,k^\prime) = \rho_{kk^\prime}$. Then $(\bL \bfSigma_{\bF}\bL^\top)_{jj^\prime} = \be_k^\top \bfSigma_{\bF} \be_{k^\prime} = \rho_{kk^\prime} = \Cor(\bY^{j},\bY^{j^\prime})$, where $\be_k$, and $\be_{k^\prime}$ denote the $k$th and $k^\prime$th standard basis vectors in $\mathbb{R}^K$, respectively. This proposition implies that if the omics variables are reordered such that those belonging to the same community $G_k$ appear consecutively, yielding the block diagonal structured $\bL$ in \eqref{eq:Loading}, then the resulting CFA-derived covariance matrix $\bL \bfSigma_{\bF}\bL^\top$ exhibits the interconnected block correlation structure.

\begin{lemma}\label{lem:independence}
   The  original predictor $\bX$ and the latent factor $\bF$ are orthogonal, denoted by $\bX \indep \bF$.
\end{lemma}
\begin{proof}
    Consider the regression model in the first step, $\bY = \bB\bX + \bE$, where $\bB$ can be estimated by the following least squares method,
\begin{align*}
    \widehat{\bB} = \argmin_{\bB \in \mathbb{R}^{p \times q}} \Vert \bY - \bX\bB \Vert_F^2.
\end{align*}
Since the OLS has the form $\widehat{\bB} = \bY \bX^\top (\bX\bX^\top)^{-1}$, the residual $\bE$ can be written as $\bE = \bY - \bY\bX^\top(\bX\bX^\top)^{-1}\bX = \bY(\bI-\bPi_{\bX})$. We can show that $\bX$ and $\bE$ are orthogonal ($\bX \indep \bE$).
\begin{align*}
    \bX \bE^\top = \bX (\bI - \bPi_{\bX}^\top)\bY^\top = \boldsymbol{0}\in\mathbb{R}^{q\times p}.
\end{align*}
Because, based on the construction of $\bF = (\bL^\top\bL)^{-1}\bL^\top\bE$, we have $\bX\indep\bF$ ($\because \bX\bF^\top = \bX \bE^\top \bL (\bL^\top \bL)^{-1} = \boldsymbol{0}\in\mathbb{R}^{q\times K}$).
\end{proof}

\begin{lemma}\label{lem:notaffect}
$\widehat{\bB}$ estimated by minimizing $\Vert\bY-\bB\bX-\bfGamma\bF\Vert_F^2$ is identical to $\widehat{\bB}$ estimated by minimizing $\Vert\bY-\bB\bX\Vert_F^2$.
\end{lemma}
\begin{proof}
 The joint estimation of $(\widehat{\bB},\widehat{\bfGamma})$ can be obtained by $(\widehat{\bB},\widehat{\bfGamma}) = \argmin_{\bB\in\mathbb{R}^{p\times q},\bfGamma\in\mathbb{R}^{p\times K}} \Vert \bY- \bB\bX - \bfGamma\bF \Vert_F^2$,
\begin{align*}
    (\widehat{\bB},\widehat{\bfGamma}) &= \bY 
    \begin{pmatrix}
     \bX^\top\\ \bF^\top   
    \end{pmatrix}
    \begin{pmatrix}
     \bX\bX^\top & \bX\bF^\top\\ 
     \bF\bX^\top & \bF\bF^\top\\
    \end{pmatrix}^{-1} \\
    &= \begin{pmatrix}\bY\bX^\top &  \bY\bF^\top\end{pmatrix} 
    \begin{pmatrix}
     (\bX\bX^\top)^{-1} & 0\\ 
     0 & (\bF\bF^\top)^{-1}\\
    \end{pmatrix}.  
\end{align*}
Thus the isolated solution for $\widehat{\bB} =  \bY \bX^\top (\bX^\top\bX)^{-1}$ is equivalent the solution without the additional predictor $\bF$.    
\end{proof}

\begin{lemma}\label{lem:joint}
Estimating $\widehat{\bB}$ and $\widehat{\bfGamma}$ by the joint procedure of $\argmin_{\bB,\bfGamma}\Vert \bY - \bB\bX - \bfGamma\bF \Vert_F^2$ is equivalent the squential estimating procedure which first estimates $\widehat{\bB}$ by  $\argmin_{\bB} \Vert \bY - \bB\bX \Vert_F^2 $ and then  $\widehat{\bfGamma}$ by $\argmin_{\bfGamma}\Vert \bE- \bGamma\bF \Vert_F^2$.    
\end{lemma}
\begin{proof}
    \Cref{lem:notaffect} showed that $\bfGamma$ is given by $\widehat{\bfGamma} = \bY\bF^\top (\bF\bF^\top)^{-1}$. Now we will show that the solution for $\widehat{\bfGamma}$ is equivalent to $\widehat{\bfGamma} = \argmin_{\bfGamma}\Vert \bE- \bfGamma\bF \Vert_F^2$, where $\bE$ is the residual after regression $\bY$ only on $\bX$, i.e., $\bE = \bY - \bY\bX^\top(\bX\bX^\top)^{-1}\bX$. 
\begin{align*}
    \widehat{\bfGamma} &= \argmin_{\bfGamma\in\mathbb{R}^{p\times K}}\Vert \bE- \bGamma\bF \Vert_F^2\\
    &= \bE \bF^\top (\bF\bF^\top)^{-1}\\
    &= (\bY - \bY\bX^\top(\bX\bX^\top)^{-1}\bX)\bF^\top (\bF\bF^\top)^{-1}\\
    &= \bY\bF^\top (\bF\bF^\top)^{-1}.
\end{align*}
\end{proof}

\begin{lemma}\label{lem:diagonal}
Assuming the factor model $\bE_i=\bL\bF_i+\bU_i , \Cov(\bE_i) = \bL \bfSigma_{\bF}\bL^\top + \bfSigma_{\bU}$ holds, then at the population level, $\Cov(\boldsymbol{\mathcal{E}}_i) = \Cov(\bE_i-\bfGamma\bF_i)$, is diagonal.  
\end{lemma}
\begin{proof}
\begin{align*}
    \Cov(\boldsymbol{\mathcal{E}}_i) &= \Cov(\bE_i-\bfGamma\bF_i)\\
    &= \Cov(\bE_i) + \Cov(\bfGamma\bF_i) - (\Cov(\bE_i,\bfGamma\bF_i)+\Cov(\bfGamma\bF_i,\bE_i))\\
    &= \bL\bfSigma_{\bF}\bL^\top + \bfSigma_{\bU} + \bfGamma\bfSigma_{\bF}\bfGamma^\top - \bL\bfSigma_{\bF}\bfGamma^\top- \bfGamma\bfSigma_{\bF}\bL^\top\\
    &=(\bL-\bfGamma)\bfSigma_{\bF}(\bL-\bfGamma)^\top + \bfSigma_{\bU}
\end{align*}
Since $\bfGamma = \bE\bF^\top (\bF\bF^\top)^{-1}$, and assuming $\bF$ is centered, $(\bL-\bfGamma)\bfSigma_{\bF}(\bL-\bfGamma)^\top$ can be expressed by
\begin{align*}
    &(\bL-\bE\bF^\top (\bF\bF^\top)^{-1})\bF\bF^\top (\bL-\bE\bF^\top(\bF\bF^\top)^{-1})^\top\\
    &=(\bL\bF-\bE\bPi_{\bF})\bF^\top (\bL-\bE\bF^\top(\bF\bF^\top)^{-1})^\top\\
    &=(\bL\bF\bF^\top-\bE\bF^\top) (\bL-\bE\bF^\top(\bF\bF^\top)^{-1})^\top\\
    &=-\bU\bF^\top (\bL-\bE\bF^\top(\bF\bF^\top)^{-1})^\top = 0.
\end{align*}    
\end{proof}

\begin{lemma}\label{lem:smallererror}
    The residual from model $\bY = \bB\bX+\bfGamma\bF$ is smaller in terms of the squared Frobenius norm than that of the model $\bY = \bB\bX+\bL\bF$, i.e., $\Vert \bY-\bB\bX-\bfGamma\bF\Vert_F^2 \leq \Vert \bY-\bB\bX-\bL\bF\Vert_F^2$.
\end{lemma}
\begin{proof}
    Due to \Cref{lem:notaffect}, including $\bF$ doesn't affect the estimation of $\bB$, and the equivalence of the joint and two-stage estimation (\Cref{lem:joint}), it suffices to compare $\Vert \bE-\bfGamma\bF\Vert_F^2$ and $\Vert \bE-\bL\bF\Vert_F^2$.
    In our CoReg model, $\bfGamma$ is obtained by OLS, i.e., $\bfGamma = \bE \bF^\top (\bF \bF^\top)^{-1}$. Thus, we can write $\Vert \bE-\bfGamma\bF\Vert_F^2 = \Vert \bE (\bI-\bPi_{\bF})\Vert_F^2$. 
$\Vert \bE - \bfGamma\bF \Vert_F^2 = \Vert \bE-\bL\bF + \bL\bF - \bfGamma\bF\Vert_F^2$
\begin{align*}
    \Vert \bE-\bL\bF\Vert_F^2-\Vert \bE - \bfGamma\bF \Vert_F^2 = &-2\tr\left((\bE-\bL\bF)^\top (\bL\bF-\bfGamma\bF) \right) -\Vert \bL \bF-\bfGamma\bF\Vert_F^2
\end{align*}
We now express the first term in $\tr$ by $(\bE-\bL\bF)^\top = ((\bE-\bE\bPi_{\bF})+(\bE\bPi_{\bF}-\bL\bF))^\top$, then by the linearity of the $\tr$, we can decompose it into two parts: (i) $\tr((\bE-\bE\bPi_{\bF})^\top (\bL\bF-\bE\bPi_{\bF}))$, and (ii) $\tr((\bE\bPi_{\bF}-\bL\bF)^\top(\bL\bF-\bE\bPi_{\bF}))$. By using the property of the projection, i.e., $(\bE-\bE\bPi_{\bF})\bF^\top = 0$, we can show the first term is $0$.
\begin{align*}
    &\tr\left((\bE-\bE\bPi_{\bF})^\top(\bL\bF-\bE\bPi_{\bF})\right) \\
    &= \tr\left((\bE-\bE\bPi_{\bF})^\top(\bL-\bE\bF^\top(\bF\bF^\top)^{-1})\bF\right)\\
    &=\tr\left((\bL-\bE\bF^\top(\bF\bF^\top)^{-1})\bF (\bE-\bE\bPi_{\bF})^\top\right) = 0
\end{align*}
Since the second term is $-\Vert \bL\bF -\bE\bPi_{\bF}\Vert_F^2$, we can simplify the above equality,
\begin{align*}
    \Vert \bE-\bL\bF\Vert_F^2-\Vert \bE - \bfGamma\bF \Vert_F^2 = \Vert \bL\bF-\bfGamma\bF\Vert_F^2 \geq 0.
\end{align*}
Thus we showed that $\Vert \bE - \bfGamma\bF \Vert_F^2 \leq \Vert \bE-\bL\bF\Vert_F^2$, where the equality holds when $\bL\bF = \bE\bPi_{\bF}$.
\end{proof}

\newpage
\section{CoReg Algorithm}

\begin{algorithm}[ht!]
\caption{Algorithm for CoReg}\label{algo:CoReg}
\begin{algorithmic}[1]
    \Step\State Fit GLM without addressing the dependence
        \begin{align*}
                \bY = \bB\bX + \bE
        \end{align*}
    \NoNumber{Estimate $\widehat{\bB}$, and $\widehat{\bfSigma} = \Cov({\bE}) = \bS \rightarrow \Cor({\bE}) = \bR$}\vspace{.5cm}
    \Step\State Coexpression Network analysis \vspace{.2cm}
    \Input $\bR$ 
      \NoNumber{\Suppressnumber\ForEach{$\boldsymbol{\lambda}=\{\lambda_1,\cdots,\lambda_l\}$} extract dense subgraphs 
        \begin{align*}
     \bigcup_{k=1}^K G_{k\vert \lambda} = \argmax_{\{\bigcup G_k, K\}} \frac{\vert \bigcup_{k=1}^K  W (G_k) \vert}{\vert \bigcup_{k=1}^K V_k\vert^\lambda}, \label{eq:Dense}
        \end{align*}
        \State{i). Estimate $\widehat{\bL}_{\lambda}$} by $\bigcup G_{k\vert \lambda}$
        \State{ii). $\widehat{\bF}_{\lambda}= \left(\widehat{\bL}_{\lambda}^\top\widehat{\bL}_{\lambda}\right)^{-1}\widehat{\bL}_{\lambda}^\top \bE$}
        \State{iii). $\widehat{\bfSigma}_{\bF_{\lambda}} = \Cov(\widehat{\bF}_{\lambda i})$}
   \EndFor\Reactivatenumber{4}}\vspace{.2cm}
   \Output{$\bF = \widehat{\bF}_{\lambda^{\ast}}$}
        \begin{align*}\text{, where }\ 
            \lambda^\ast = \argmin_{\lambda\in\boldsymbol{\lambda}}\Vert \widehat{\bfSigma} - \widehat{\bL}_{\lambda}\widehat{\bfSigma}_{\bF_{\lambda}}\widehat{\bL}_{\lambda}^\top \Vert_2^2 + \Vert \widehat{\bL}_{\lambda} \Vert_\ast
        \end{align*}
\end{algorithmic}
\begin{algorithmic}[1]
\Reactivatenumber{3}
    \NewStep\State {Fit CoReg model with $\bF$ as predictors}
        \begin{align*}
            \bY = \bB\bX + \bfGamma\bF + \boldsymbol{\mathcal{E}}
        \end{align*}
\end{algorithmic}
\end{algorithm}
\FloatBarrier

\newpage
\section{Comparison of Community Detection Methods}

We compared CoReg with other commonly used community detection methods, including Weighted Gene Co-expression Network Analysis (WGCNA) and the Louvain algorithm in \Cref{fig:Comparison_community}.  The resulting organized correlation heatmaps demonstrate that CoReg successfully reveals the underlying interconnected community structure, whereas the other methods fail to capture these patterns.

\begin{figure}
    \centering
        \includegraphics[width=\textwidth]{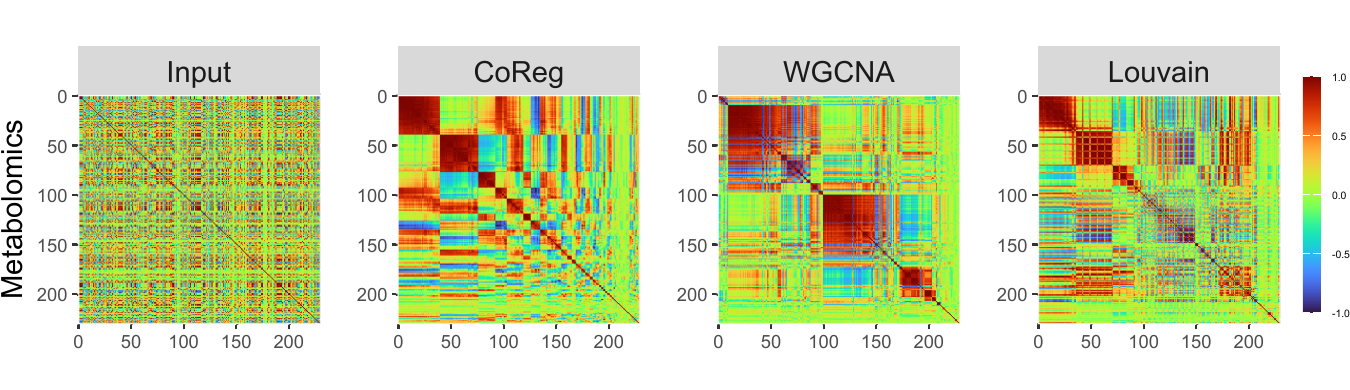}
        \includegraphics[width=\textwidth]{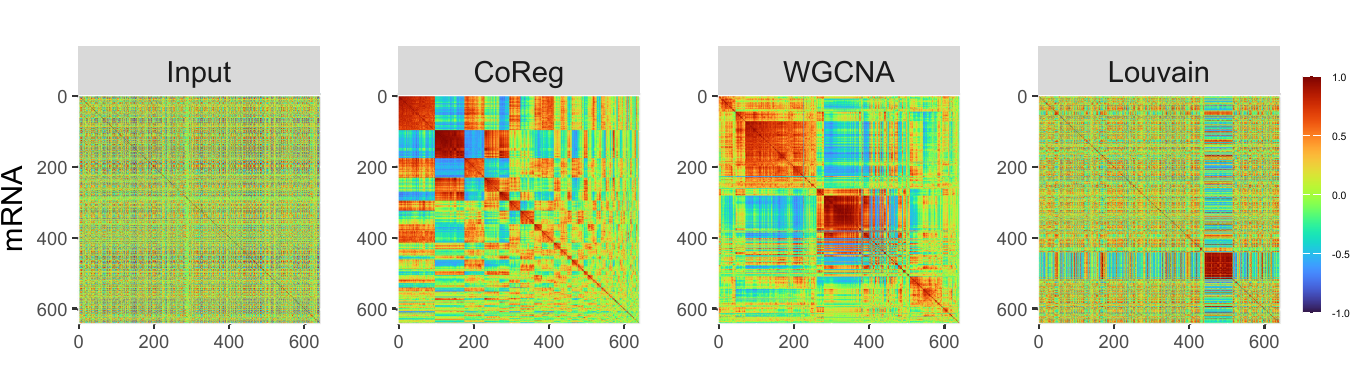}
        \caption{Comparison of community detection methods between CoReg and other commonly used approaches, including Weighted gene co-expression network analysis (WGCNA) and the Louvain algorithm. Top: ADNI metabolomics data ((C) of \tblue{Figure 1} in the main manuscript), Bottom: Leukocyte Messenger RNA (mRNA) expression data ((B) of \tblue{Figure 1} in the main manuscript). The organized correlation heatmaps demonstrate that CoReg successfully reveals the underlying interconnected community structure, whereas the other methods fail to capture these patterns.}
        \label{fig:Comparison_community}
\end{figure}
\FloatBarrier


\section{Additional simulation results}

This section includes additional simulation results.

\subsection{Computational Efficiency Analysis}

The computational complexity of CoReg is determined by co-expression network analysis, CFA, and multivariate regression analysis. As we employ a greedy peeling algorithm for module extraction, the computational complexity of this step is bounded by $\mathcal{O}(p^2)$ under the assumption of all edges being densely connected. Under a commonly used random graph assumption $G(p,\pi)$ with connection probability $\pi=c \frac{\log p}{p}$, the expected computational complexity reduces to $\mathcal{O}(p \log p)$. The computation cost for calculating the factor score has a complexity of  $\mathcal{O}(nK^2+K^3)$,  where $K$ is the size of the extracted subgraph, and $K \ll p$. Combining these results, the overall complexity is $\mathcal{O}(p \log p+nK^2+K^3)$, which is computationally efficient. In contrast, the computational complexity of SVA, which is $\mathcal{O}(pn^2+n^3)$, stems from the singular value decomposition of the $p\times n$ matrix.

We further demonstrate the computational cost of CoReg using simulation studies with varying sample sizes (i.e., $n$) and dimensions of the response variable (i.e., $p$). For varying $n=\{100,1000\}$ we fixed $p=200$, while for varying $q=\{100,2000\}$, we fixed $n=200$. Computation times were measured in seconds on a system equipped with an Intel Core i7-8850H CPU at 2.60 GHz and 16 GB of RAM. The mean computation times averaged over 20 replications for each setting are displayed in \Cref{fig:Computation_Time}. As both $n$ and $p$ increase, the proposed CoReg demonstrates superior computational efficiency, exhibiting only marginal increases in time consumption, while SVA requires increased computational time.

\begin{figure}
    \centering
    \includegraphics[width=\textwidth]{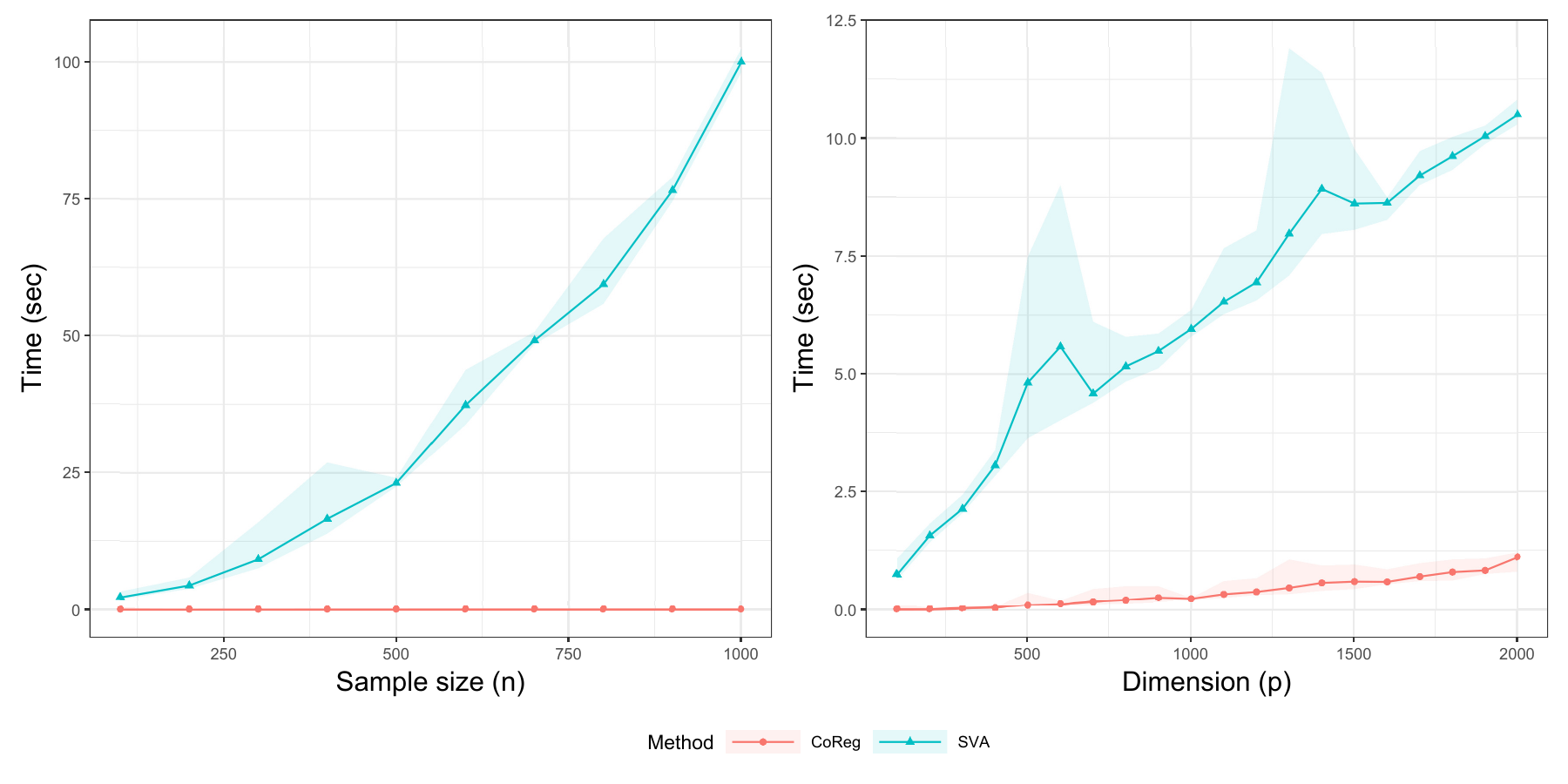}
    \caption{Computation time, measured in seconds and averaged over 20 replications. Upper panel: The sample size ($n$) was varied while maintaining a fixed dimension ($p=200$). Bottom panel: The dimension of outcome variables ($p$) was varied while the sample size remained fixed ($n=100$).}
    \label{fig:Computation_Time}
\end{figure}
\FloatBarrier


\subsection{Supplementary figures for simulations}

Heatmaps of correlations from simulated data, ROC curves, $p$-value histograms, and the Venn diagram from the replicability analysis (Scenario 2) are displayed in \Cref{fig:Sim_Heatmap}, \Cref{fig:p_val_histo},  and \Cref{fig:Replicability_scenario2}, respectively.

\begin{figure}
    \centering
    \includegraphics[width=\textwidth]{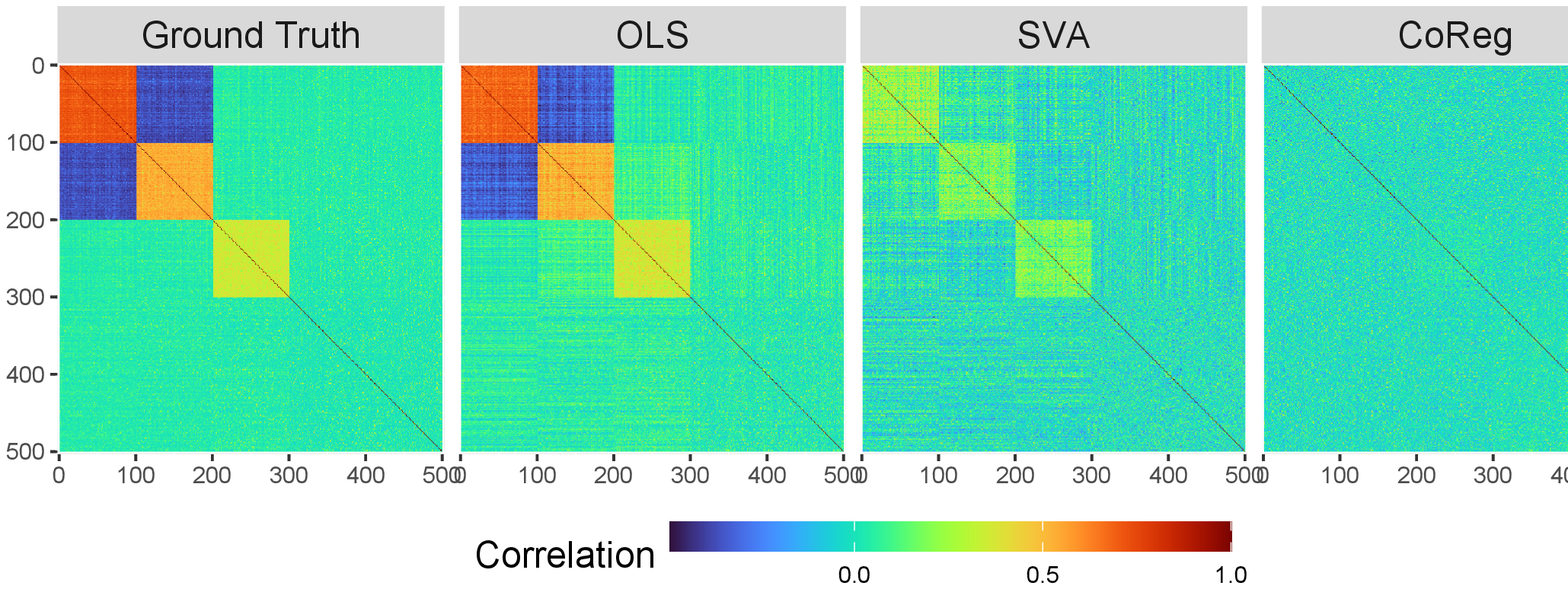}
    \caption{Heatmap of Correlations: The Ground Truth heatmap displays the correlation structure among the response variables $\bY$, featuring three densely correlated diagonal blocks:  the first block in the left corner shows strong correlations in red, the second block shows moderate correlations in orange, and the third displays weaker correlations in yellow. Negative inter-block correlations are represented in blue on the off-diagonals. In the residual heatmap from the OLS model, a structured correlation pattern remains, suggesting that OLS fails to address dependence. Although SVA addresses the correlation among variables, it does not fully eliminate dependencies, as evidenced by the yellow-colored diagonal blocks. In contrast, the heatmap of CoReg appears almost uniformly green, indicating that it effectively addresses the structured correlation pattern.}
    \label{fig:Sim_Heatmap}
\end{figure}
\FloatBarrier


\begin{figure}
    \centering
    \includegraphics[width=\textwidth]{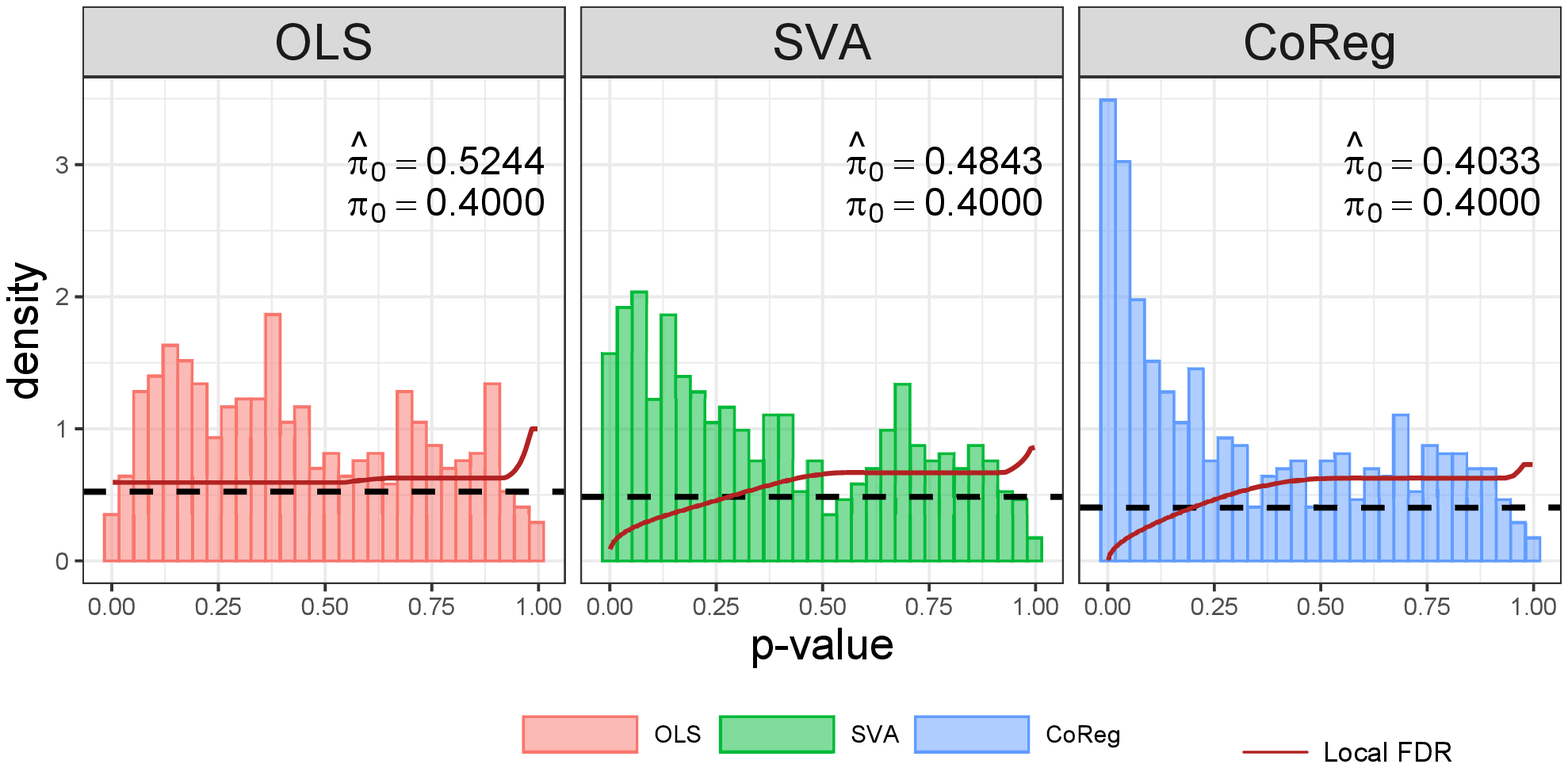}
    \caption{$p$-value histogram from simulation study. }
    \label{fig:p_val_histo}
\end{figure}
\FloatBarrier

\begin{figure}
    \centering
    \includegraphics[width=\textwidth]{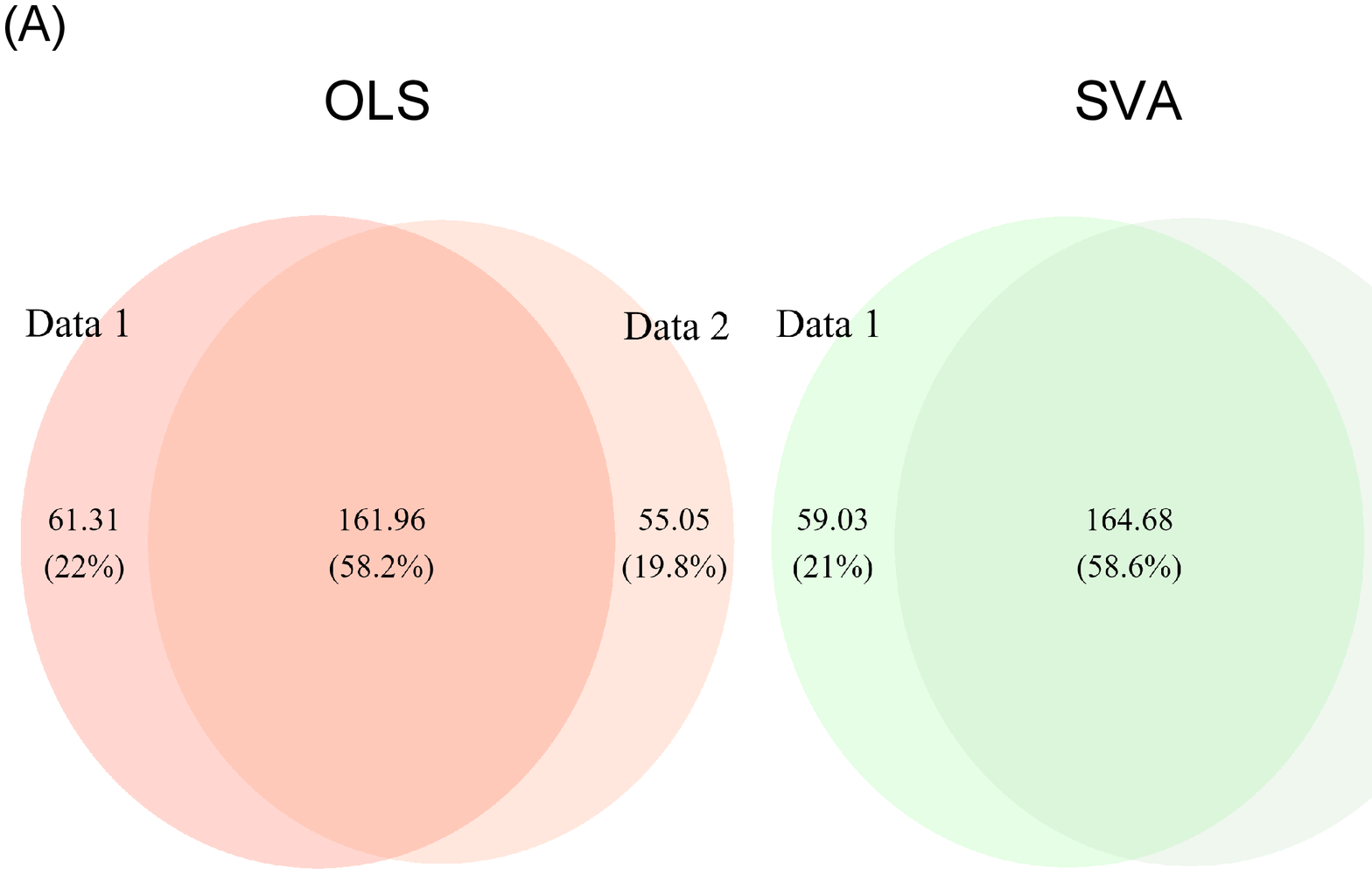}
    \caption{Result from the replicability analysis (scenario 2): Two datasets were generated under different experimental conditions, with each dataset having the same noise level but different effect sizes. We simulated 100 replications. The values represent the average true positive count for each section in the Venn diagram, with proportions provided in parentheses. CoReg demonstrated a larger intersection compared to other competing methods, indicating superior replicability.}
    \label{fig:Replicability_scenario2}
\end{figure}
\FloatBarrier

\section{Additional omics data analysis results}

This section includes additional omics data analysis results. \Cref{Tab:Result_ADNI_UKB} presents the statistical inference results for the omics data. \Cref{fig:Real_Heatmap} displays the correlation heatmap. \Cref{fig:ADNI_Pvalues} shows the $p$-value histograms for the ADNI data obtained using various methods. \Cref{fig:UKBB_subsample} and \Cref{fig:Real_UKBB_speed} show the results of subsample analyses conducted on the UKBB data.

In \Cref{fig:ADNI_Pvalues}, the raw $p$-value density histogram of CoReg is shown in comparison to those of competing methods. The CoReg $p$-value distribution exhibits the different shape showing a improved sensitivity, with a pronounced peak near 0 and a relatively uniform tail for larger $p$-values \citep{StoreyTibshirani:2003}, which facilitates easier separation of significant findings. In contrast, $p$-values from the other methods are uniformly distributed across the entire range, lacking the distinct concentration near 0 that indicates no positive findings after multiple testing corrections. We further estimated the proportion of true null hypotheses ($\widehat{\pi}_0$) for all tested hypotheses following \cite{Storey:2002qvalue} and computed local FDR values \citep{Efron:2001localfdr}. Consistently lower local FDR at low $p$-values and low $\widehat{\pi}_0$ from CoReg, compared to those from other methods, suggests higher sensitivity and more potential true discoveries, whereas other methods appear more conservative, potentially missing true signals.  The improved sensitivity is well aligned with the findings in the simulation study.

\begin{table}
\centering
    \caption{Results from the ADNI DNA methylation data (top). The numbers indicate the count of CpG sites with a statistically significant effect of age and AD Diagnosis detected using each method. The values in parentheses indicate the ratio of significant CpG sites to the total number of CpG sites analyzed ($q = 4200$). In UKBB metabolomics data (bottom), the analysis assessed the associations of metabolic profiles with age and FR-PDFF.}
 \label{Tab:Result_ADNI_UKB}
    \begin{tabular}{lrrrrr}
    \toprule
    \multicolumn{4}{c}{ADNI DNA Methylation data ($n=364, p=4200$)}\\
    \midrule
        &  OLS & SVA & CoReg  \\
    Time (sec.) & $\cdot$& 108.34  & 13.62 \\
    Age & 32 (0.76\%) & 10 (0.24\%) & 543 (12.93\%)\\     
    AD Diagnosis &0 (0.00\%) & 0 (0.00\%) & 54 (1.29\%)\\         
    \midrule
    \multicolumn{4}{c}{UKBB Metabolomics data ($n=7329, p=249$)}\\
        \midrule
        &  OLS & SVA & CoReg  \\
    Time (sec.) & $\cdot$& 3724.31  & 0.43 \\
    Age & 229 (91.97\%) & 229 (91.97\%) & 245 (98.39\%)\\
    FR-PDFF & 199 (79.92\%) & 199 (79.92\%) & 240 (96.39\%)\\
    \bottomrule
    \end{tabular}
\end{table}
\FloatBarrier

\begin{figure}
    \centering
    \includegraphics[width=\textwidth]{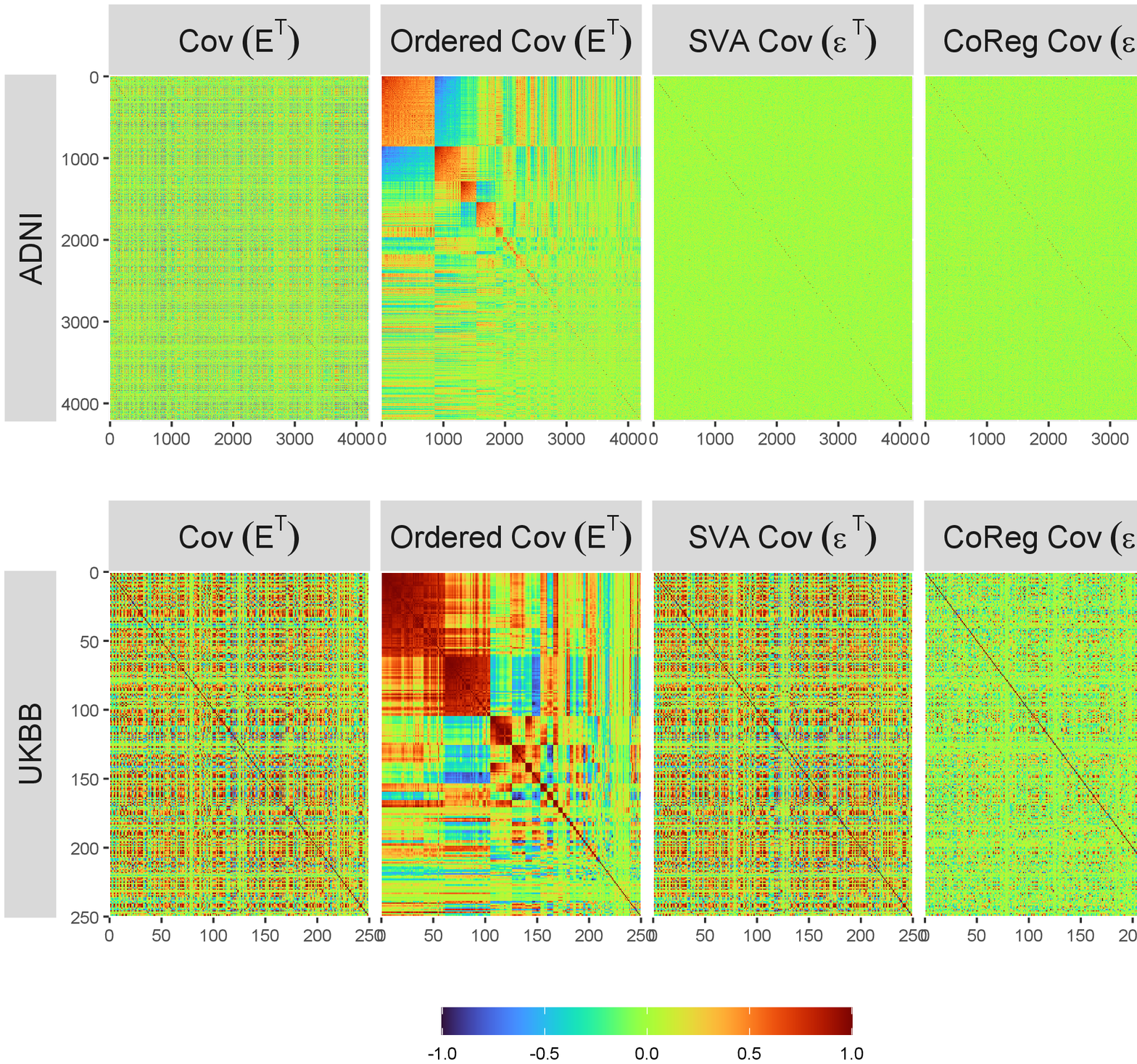}
    \caption{Heatmaps of residual correlations from two datasets (Top: ADNI, Bottom: UKBB) are presented. Methods that more effectively model the dependence show less dependence, as indicated by a more uniform (green) heatmap. In contrast, methods that fail to capture the dependence display highly scattered positive (red) and negative (blue) elements in the heatmap.}
    \label{fig:Real_Heatmap}
\end{figure}
\FloatBarrier

\begin{figure}
    \centering
    \includegraphics[width=\textwidth]{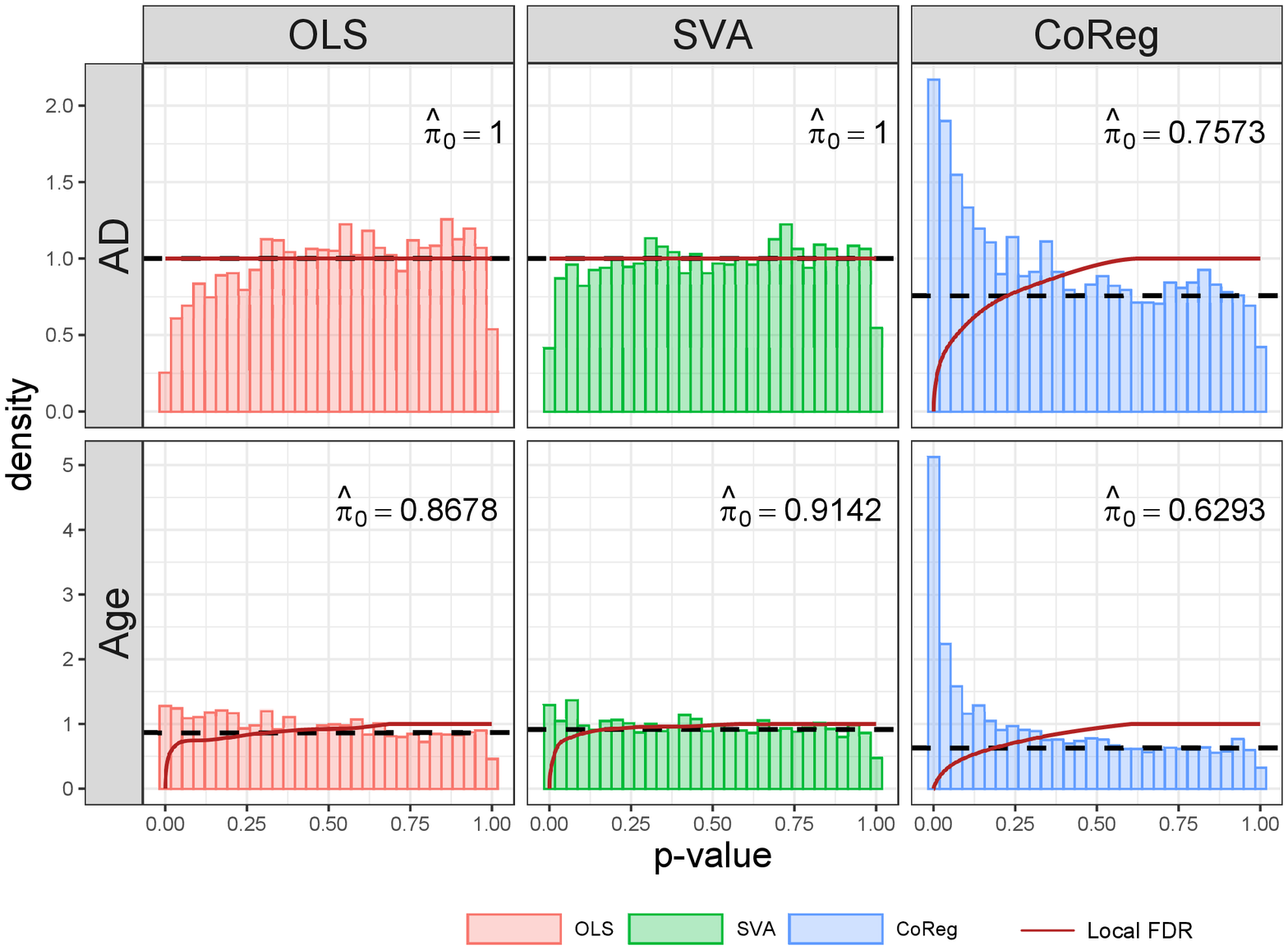}
    \caption{Histogram of raw $p$-values from the ADNI DNA methylation dataset, illustrating the distribution of $p$-values for the age effect (top) and Alzheimer's disease effect (bottom) on DNA methylation levels. The estimated proportion of true null hypotheses ($\widehat{\pi}_0$) for all tested hypotheses based on \cite{Storey:2002qvalue} is provided, with the corresponding black dashed line indicating the value of $\widehat{\pi}_0$. The red solid curve represents the local FDR \citep{Efron:2001localfdr} for each test based on $p$-values.}
    \label{fig:ADNI_Pvalues}
\end{figure}
\FloatBarrier

\begin{figure}
    \centering
    \includegraphics[width=\textwidth]{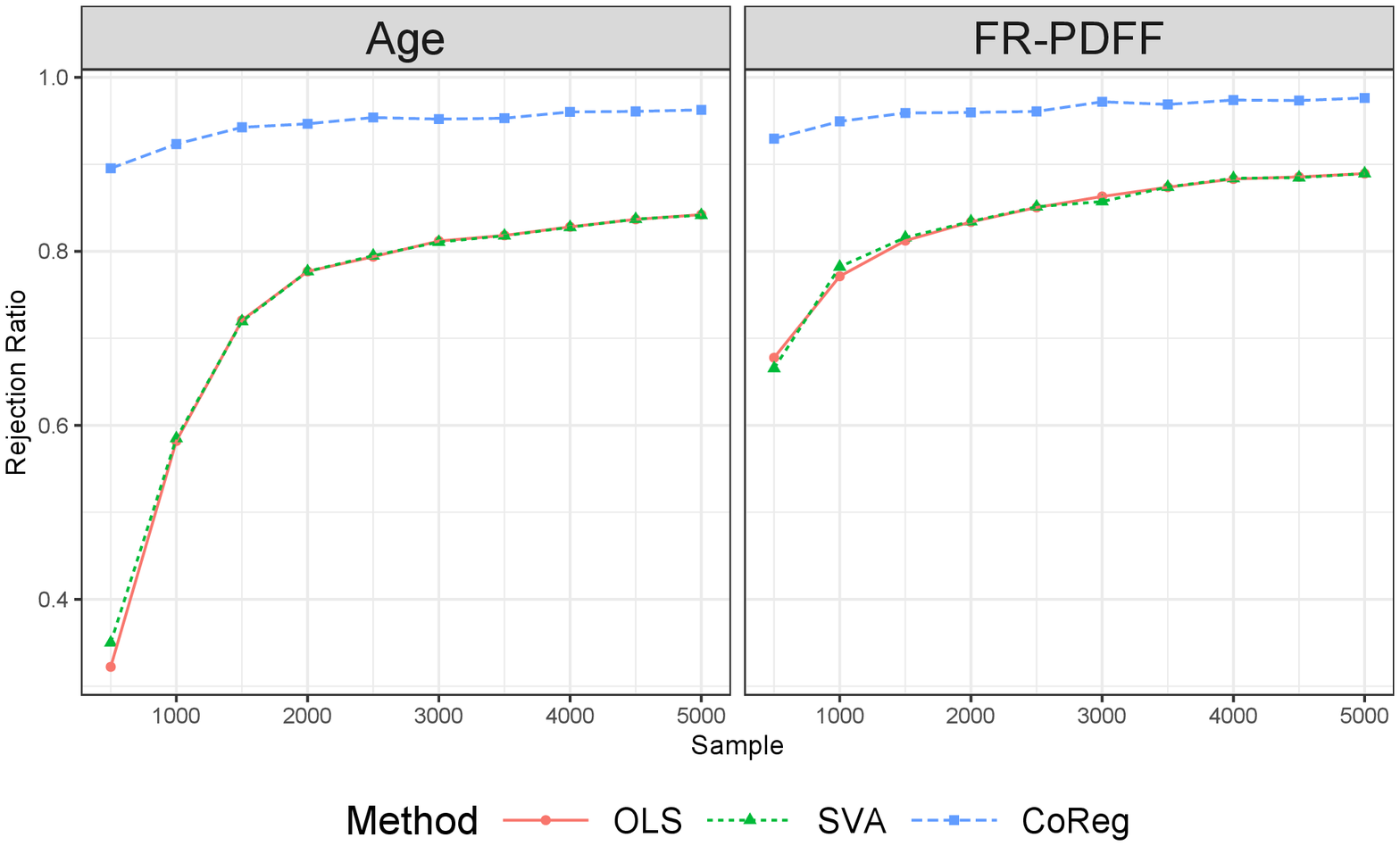}
    \caption{Performance analysis using a subset of the data. The rejection ratio, averaged over 20 replications, represents the proportion of tests rejected (indicating a significant effect) among 249 metabolite measurements. The $x$-axis denotes varying sample sizes, illustrating how the rejection ratio changes with sample size.}
    \label{fig:UKBB_subsample}
\end{figure}
\FloatBarrier

\begin{figure}
    \centering
    \includegraphics[width=\textwidth]{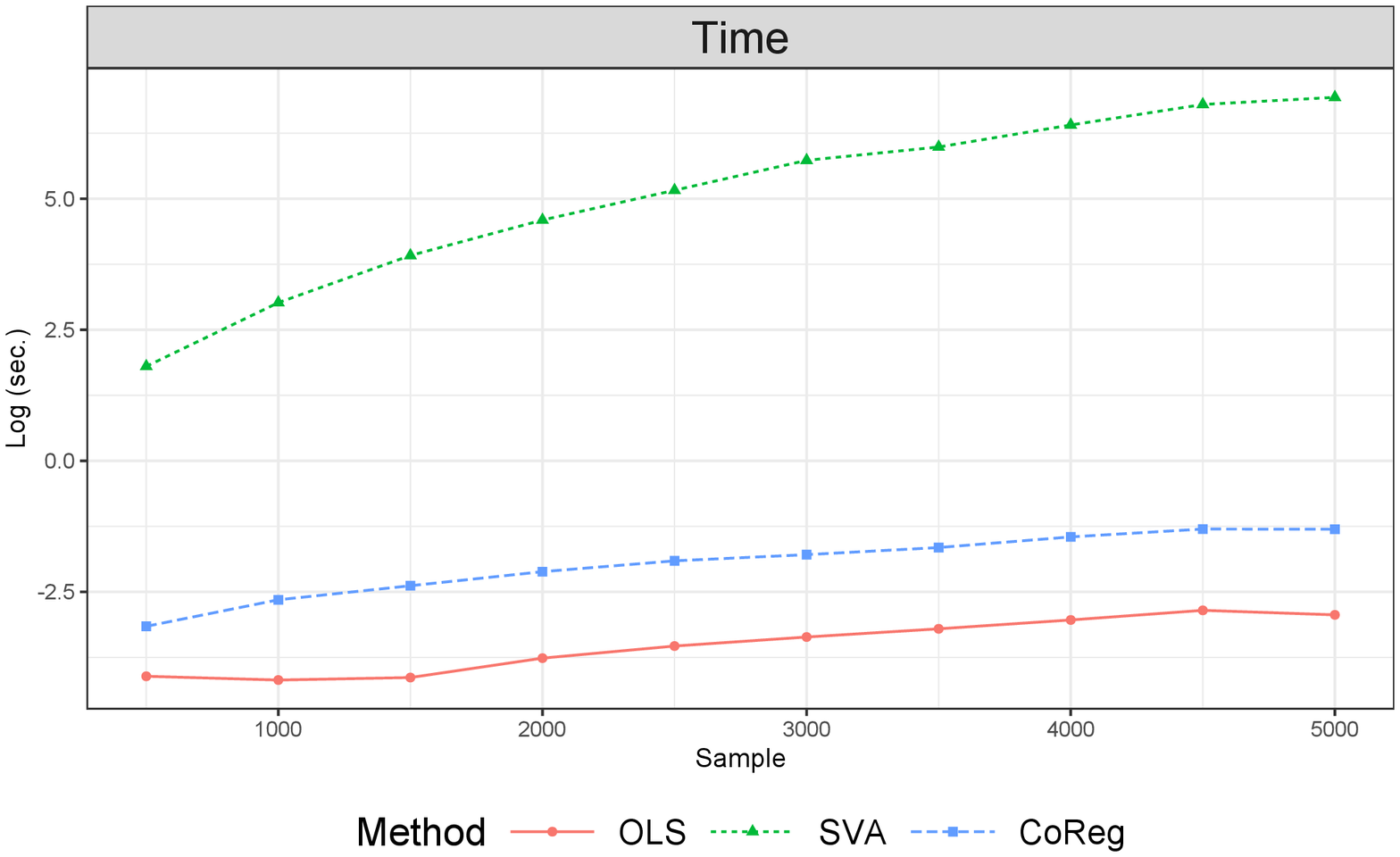}
    \caption{Computational time comparison from different methods. }
    \label{fig:Real_UKBB_speed}
\end{figure}
\FloatBarrier








\end{document}